\RequirePackage{fix-cm}
\RequirePackage{amsmath}
\documentclass[smallextended]{svjour3}       
\smartqed  
\usepackage{graphicx}
 \journalname{Mathematical Programming}
\usepackage{amsthm}
\usepackage{amsfonts, mathdots, amssymb, latexsym, amsbsy, bm, mathtools}   
\usepackage{braket}
\usepackage{dsfont}

\usepackage[unicode]{hyperref} 
\makeatletter
\let\cl@chapter\undefined
\makeatletter
\usepackage{cleveref}

\usepackage{thmtools}
\usepackage{thm-restate}

\usepackage{color}

\newcommand{\R}{\mathbb{R}}
\newcommand{\N}{\mathbb{N}}

\newcommand*\diff{\mathop{}\!\mathrm{d}}

\newcommand{\abs}[1]{\left\lvert #1 \right\rvert}

\renewcommand{\epsilon}{\varepsilon}

\renewcommand{\l}{\ell} 
\newcommand{\ext}{\alpha} 
\newcommand{\fac}{c} 
\newcommand{\capo}{\nu^-} 
\newcommand{\capi}{\nu^+} 
\newcommand{\capmin}{\nu_{\min}} 
\renewcommand{\b}{b} 
\newcommand{\bwave}{\tilde{b}} 
\newcommand{\dwave}{\tilde{d}} 
\newcommand{\Ewave}{\tilde{E}} 
\newcommand{\etamin}{\eta_{\min}}
\newcommand{\Ve}{V}
\newcommand{\Vz}{\bar V}

\newcommand{\vz}{\bar v}

\newcommand{\paths}{\mathcal{P}}

\DeclareMathOperator{\VI}{VI}
\DeclareMathOperator{\SOL}{SOL}

\DeclareMathOperator*{\argmax}{arg\,max}

\renewcommand\labelenumi{(\roman{enumi})}
\renewcommand\theenumi\labelenumi

\title{Nash Flows Over Time with Spillback and Kinematic Waves\thanks{This research was funded by the Deutsche Forschungsgemeinschaft (DFG, German Research Foundation) under Germany's Excellence Strategy The Berlin Mathematics Research Center MATH+ (EXC-2046/1, project ID: 390685689) and the Research Training Group 2236 UnRAVeL.}
}
\author{Leon Sering        \and
        Laura Vargas Koch 
}

\institute{L. Sering (ORCID: 0000-0003-2953-1115) \at
              Stra\ss{}e des 17. Juni 136, 10623 Berlin, Germany  \\
              \email{sering@math.tu-berlin.de}           
           \and
           L. Vargas Koch (ORCID: 0000-0002-7499-5958) \at
             Kackertstr. 7, 52072 Aachen, Germany\\
             \email{laura.vargas@oms.rwth-aachen.de}
}

\date{}
\begin{document}

\maketitle

\begin{abstract}Modeling traffic in road networks is a widely studied but challenging problem, especially under the
assumption that drivers act selfishly. A common approach is the \emph{deterministic queuing model}, for which the
structure of dynamic equilibria has been studied extensively in the last couple of years. The basic idea is to model
traffic by a continuous flow that travels over time through a network, in which the arcs are endowed with transit times
and capacities. Whenever the flow rate exceeds the capacity the flow particles build up a queue. So far it was not
possible to represent spillback or kinematic waves in this model. By introducing a storage capacity arcs can become
full, and thus, might block preceding arcs, i.e., spillback occurs. Furthermore, we model kinematic waves by upstream
moving flows over time representing the gaps between vehicles. We carry over the main results of the original model to
our generalization, i.e., we characterize \emph{Nash flows over time} by sequences of particular static flows, so-called
\emph{spillback thin flows}. Furthermore, we give a constructive proof for the existence of dynamic equilibria, which
suggests an algorithm for their computation. This solves an open problem stated by~\cite{koch2010nash}.

\end{abstract}
 \section{Introduction} \label{sec:intro}
 
Urban population is rapidly growing worldwide and so is the number of vehicles in metropolitan areas. To get control of
this rising traffic volume intelligent traffic planning is of central importance. In particular it is essential to solve
many of the major traffic problems in today's cities, e.g., air and noise pollution and long travel times. A well planned traffic does not only increase the quality of life for traffic users but also benefits the economy
and environment. Improved navigation systems and the availability of massive amounts of traveling data give a huge
opportunity to optimize the infrastructure for the growing demand. This draws the attention to more realistic
mathematical traffic models and algorithmic approaches for the interplay of individual road users. %
Unfortunately, on the one hand, realistic models used in simulations are mathematically poorly understood and, on the
other hand, theoretically precise models that are mathematically well-analyzed are very simplified. Our contribution is
to extend the theoretical state of the art model by adding two crucial components: \emph{spillback} and \emph{kinematic waves}. These effects can be
observed in daily traffic situations, e.g., spillback occurs on a highway, where a bottleneck causes a long traffic jam that blocks exits
upstream, or during rush hour in a big city where a crossing is impassable due to the congestion of an intersecting
road. 
Furthermore, traffic congestions are observed to move upstream after a bottleneck is removed as vehicles have a certain reaction time to close the gap if a preceding car accelerates. In other words, the gaps between vehicles move backwards over time, which causes a wave-like motion of the congestion.
This phenomenon is therefore called \emph{kinematic waves} and it was first studied from a mathematical perspective by \cite{lighthill1955kinematic,lighthill1955kinematic2}.
It is no surprise that spillback and kinematic waves are of great interest for traffic planners and that these are core features of recent
traffic simulation tools.
Hence, introducing spillback and kinematic waves is an important step towards closing the gap between
mathematical models and simulations.
%

 \paragraph{Flows over time.}
 As there is a huge number of interacting agents in daily traffic, we do not concentrate on single entities but consider
 traffic streams instead. For this scenario, in which infinitesimally small agents travel through a network over time,
 \emph{flows over time} are an excellent mathematical description. While the game theoretical perspective of this
 problem is still in its infancy, the optimization perspective has already been studied for more than half a century. \cite{ford1958constructing} introduced a time-dependent flow model, in which flow travels over
 time through a network from a source~$s$ to a sink~$t$. Every arc of the network is equipped with a capacity, which
 limits the rate of flow using that arc, and a transit time specifying the time needed to traverse it. This model is
 widely analyzed and there are several algorithms solving different optimization problems. Ford and Fulkerson presented
 an algorithm for the maximum flow over time problem, i.e., sending as much flow as possible from~$s$ to~$t$ within a
 given time horizon. A natural extension is to search for flows over time that maximize the flow amount reaching the sink for every
 point in time simultaneously, so called \emph{earliest arrival flows}. In 1959, \cite{gale1959transient} proved
 their existence in an $s$-$t$-network and the first algorithm was presented by
 ~\cite{wilkinson1971algorithm}. 
All these problems were first considered from a discrete time perspective and only in 1996, 
\cite{fleischer1998efficient} showed that all the results and algorithms carry over to the continuous time model, which
has become the conventional perspective by now. For a nice introduction into the whole field we refer to the survey
of~\cite{skutella2009introduction}.

\paragraph{Dynamic equilibria.} Meanwhile flows over time were considered from a decentralized, game theoretical
perspective in the transportation science community; see, e.g., the book of~\cite{ranboyce96} and the article about spillback of~\cite{daganzo98spillover}. In accurate traffic scenarios it is reasonable
to expect the participants (particles) to act selfishly, i.e., to minimize their arrival times. The actual traffic is
then represented by a dynamic equilibrium, a state where no particle can reach the destination quicker by changing
its route. In this paper we consider the \emph{deterministic queuing model} to describe the arc dynamics, which is also
used in the simulation software MATSim; see \cite{horni2016multi}. In this competitive flow over time setting it is possible
that the inflow rate exceeds the capacity for some arc, which causes a queue to build up in front of the exit.
Therefore, the actual travel time of an arc consists of the transit time plus the queue waiting time. \cite{koch2010nash} characterized the structure of dynamic equilibria, called \emph{Nash flows over time}, and
showed that they consist of a number of phases, in which the in- and outflow rates of each arc are constant. Each
phase is characterized by a particular static flow together with node labels, named \emph{thin flows with resetting}.
\cite{cominetti2011existence} showed the existence and uniqueness of such thin flows, from which a Nash flow over time can be constructed and \cite{cominetti2015dynamic} extended the
existence result to networks with general inflow rate functions and to a multi-commodity setting. Moreover, \cite{cominetti2017long} examined the long term behavior of queues and were able to bound
their lengths whenever the network capacity is sufficiently large. Finally,
\cite{bhaskar2015stackelberg} analyzed different prices of anarchy in this model and \cite{sering2018multiterminal} showed how to translate the constructive results to a multi-terminal setting.

\paragraph{Our contribution.} In the \emph{Koch-Skutella-model} introduced by \cite{koch2010nash} the queues do not have
physical dimensions and can in principle be arbitrarily large. Thus, spillback cannot occur, which is a huge drawback
when considering real world scenarios. Our contribution is to extend this model such that the total amount of flow on an
arc, and thus the queue length, can be bounded. Whenever an arc is \emph{full} the inflow rate cannot exceed the outflow
rate anymore. In words of traffic: if a road is full no new vehicle can enter the street before another vehicle leaves.
If more flow aims to use a full arc, it has to queue up on a previous arc, i.e., we have \emph{spillback}. We generalize
the concept of thin flows to this spillback setting by introducing an additional node label, which we call
\emph{spillback factor}. We show that, similar to the Koch-Skutella-model, the derivatives of every Nash flow over time
with spillback form a \emph{spillback thin flow}. In reverse it is possible to compute a Nash flow over time by
extending a flow over time step by step via spillback thin flows.
We generalize this model further to capture kinematic waves.
Whenever flow leaves an arc, it takes time until the free space reaches the tail of the arc. We model this by introducing a backwards moving flow over time on each arc representing these gaps. We show that all
results for the spillback model transfer to this extension and that it satisfies the conditions for kinematic wave
models stated by the traffic scientists; see \cite{flotterod2016queueing}.

\paragraph{Outline.} In \Cref{sec:street} we motivate the spillback model via road properties and give an illustrative
example to emphasize the importance of spillback. \Cref{sec:flow_dynamics} introduces the basic notations and
concepts of the flow dynamics. In \Cref{sec:Nash_flows} we define Nash flows over time and spillback thin flows
and show their structural connection. \Cref{sec:constructing_nash_flows} is dedicated to the construction and
computation of Nash flows over time via spillback thin flows. In \Cref{sec:kinematic_wave_model}, we introduce the kinematic wave model, again by a motivation via road properties, and then show that the results for the spillback model transfer to this extension. In \Cref{sec:koch_skutella} we present the connection to already known results, especially regarding uniqueness and the price of anarchy.
Finally, in \Cref{sec:outlook}, we give a brief conclusion
and outlook on further interesting questions. Due to space restrictions some of the more technical proofs are moved to the appendix (\Cref{sec:appendix}).  

\section{From Roads to Arcs}  \label{sec:street} \paragraph{Street model with spillback.} In order to get an appropriate model with spillback (without kinematic waves for now) imagine
a street with a number of lanes~$w$, a length~$\l$, and a speed limit~$v_1$; see on the left of \Cref{fig:street_model}.
We assume the street ends in a crossing, which leads to an exit speed limit denoted by~$v_2$. The number of cars
entering or leaving the street per time unit is denoted by $f^+(\theta)$ and $f^-(\theta)$ respectively. Whenever some
cars cannot leave the road immediately a traffic jam, also called \emph{physical queue}, builds up at the end of the road. Its length at time~$\theta$ is
denoted by~$j(\theta)$ and if it equals $\l$, the street is full. A new car can only enter the street if there is enough
space on a lane at this moment. After entering it drives along the street with velocity $v_1$ until it reaches either
the end of the street or the end of the traffic jam. In the latter case, it stays in the stop-and-go traffic until it
reaches the end of the street. Whenever the following street is full the outflow is restricted, therefore, the stop-and-go
speed is reduced even further leading to longer traffic jams. This is what we call spillback.
\begin{figure}[ht]
\centering 
\includegraphics[width=\columnwidth]{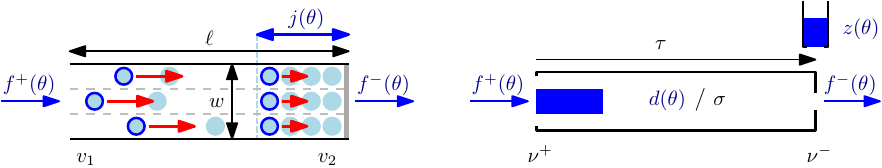} 
\caption{\emph{Left:} Model of cars using a street with~$w$~lanes, length~$\l$, speed
limit~$v_1$, and exit speed~$v_2$. \emph{Right:} Simplification we use within this paper to model a street. 
} 
\label{fig:street_model}
\end{figure}
\paragraph{Arc model.} To describe this situation mathematically we consider a directed graph. Hereby, each arc
corresponds to a street segment and every node to a crossing. In order to make the street dynamics easier to handle we
transform the  properties of the street in the following way; see on the right of \Cref{fig:street_model}: Each arc is
equipped with an inflow capacity~$\capi$ corresponding to~$w \cdot v_1$, a~free flow transit time~$\tau$ corresponding
to~$\frac{\l}{v_1}$, a storage capacity~$\sigma$ corresponding to~$\frac{\l \cdot w}{\text{car-length}}$ and an outflow
capacity~$\capo$ corresponding to~$v_2 \cdot w$. Instead of considering discrete cars we look at a continuous flow over
time that is described by an inflow rate function~$f^+$ and an outflow rate function~$f^-$. In this flow model the queue
does not have a physical length, i.e., each flow particle first traverses the arc in $\tau$ time and if the \emph{point
queue}~$z(\theta)$ is positive the particle lines up. The in- and outflow
rate are the same whether the queue has a physical length as described by the traffic jam in the street model or
consists of a point queue, since in both cases we have strict FIFO. The total amount of flow on an arc, consisting of
traversing particles and the flow in the queue, is denoted by~$d(\theta)$ and can never exceed the storage
capacity~$\sigma$. Note that the amount of flow in a physical queue is in general larger than the
amount of flow in the point queue, which is the reason why we restrict $d(\theta)$ by the storage capacity and not
$z(\theta)$. In fact, $d(\theta) = \sigma$ corresponds exactly to the case that the length of the traffic jam equals the
length of the street, i.e., the street is full. 

\paragraph{Introductory example.} To illustrate the importance of spillback we present two examples of Nash flows over
time. We consider the same network in both cases except for a different outflow capacity on arc~$e_2$. In the first
case, depicted on the left of \Cref{fig:example}, suppose that the outflow capacity is~$\capo_{e_2} = 1$. Since the
unique shortest path in the network uses arc~$e_2$ all flow particles use this path until time~$4$. At this point in
time particles located at node~$v$ that decide to use~$e_2$ will be at the end of this arc at time~$5$ and will
experience a queue of length~$6$. Hence, the total travel time along~$e_2$ is~$7$, which equals the transit time
of~$e_3$. Thus, the flow splits up: a rate of $2$ takes arc~$e_3$ and the remaining flow of rate $1$ chooses arc~$e_2$
nevertheless. Since the inflow rate of arc~$e_2$ is now~$1$, and therefore equal to the outflow rate, the queue length
stays constant at~$6$ and the total travel time from~$s$ to~$t$ remains constant at $8$ for all times.

\begin{figure}[ht]
\begin{center}
\includegraphics[width=\columnwidth]{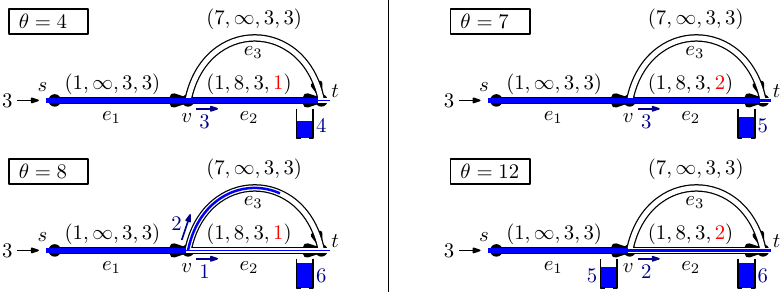} 
\caption{The arc properties are displayed in the following order $(\tau_e,\sigma_e, \capi_e,\capo_e)$. 
} 
\label{fig:example}
\end{center}
\end{figure}

For the second case, depicted on the right of \Cref{fig:example}, we consider the same network except for the outflow
capacity $\capo_{e_2} = 2$, which changes the situation drastically. As before, in the beginning the path along $e_2$ is
the unique shortest route, and thus, all particles take it. At time $7$ the arc gets full because the amount of flow in
the queue equals $5$ and $3$ additional flow units are traversing the arc at this moment. From this time onward the
inflow rate into $e_2$ is restricted to the outflow rate, namely $2$. Note that for particles located at $v$ at time~$7$
the travel time along $e_2$ equals $4$: one time unit for traversing the arc and a waiting time of $3$ (at time $8$
the amount of flow in the queue is $6$ and the particles leave the queue with a rate of $2$). Hence, for these particles
$e_2$ is still faster than~$e_3$ and since the queue will never become longer than $6$ all later particles will stick to
the bottom route. Since $e_2$ is full, they spill back and queue up on arc $e_1$ from time $7$ onward. It follows that
the total travel time from $s$ to $t$ will rise unbounded for later particles.

This example points out that the storage capacity might have a huge influence on the dynamic equilibrium and it shows
that this can even lead to counter-intuitive dynamics, since widening the capacity on $e_2$ leads to a longer travel time
for later particles.

  \section{Spillback Model} \label{sec:flow_dynamics} In this section we introduce the dynamic queuing model with
  spillback, i.e., we specify the properties of the network and the flow dynamics on the arcs. Note that the model is a
  generalization of the model introduced in~\cite{koch2010nash} (see \Cref{prop:generalization}) and the structure of this article follows the lines
  of~\cite{cominetti2011existence,cominetti2015dynamic,koch2010nash}.
  Throughout this paper we consider a directed graph~$G = (V, E)$ with transit times~$\tau_e \geq 0$, outflow
  capacities~$\capo_e > 0$, inflow capacities~$\capi_e > 0$ and storage capacities~$\sigma_e \in (0, \infty]$ on every
  arc~$e \in E$.
  Furthermore, there are two distinguished nodes, a source~$s \in V$ with an inflow rate~$r \geq 0$ and a sink~$t \in
  V$. We assume that every node is reachable from~$s$ and that there is no directed cycle with zero transit time. In
  order to ensure that traversing flow alone can never fill up the storage of an arc $e$ we require that $\sigma_e >
  \capi_e \cdot \tau_e$. We assume that $\sigma_e = \infty$ and $\capi_e > r$ for all $e \in \delta^+(s)$
  and~$\delta^-(s)=\emptyset$ to ensure that spillback never reaches the source, and thus, the network inflow is
  never throttled. This is without loss of generality, because we can ensure these requirements by adding a new source~$s^*$ and a new arc~$e^* = s^*s$ with~$\tau_{e^*} = 0$, $\sigma_{e^*} = \infty$,
  $\capo_{e^*} = r$, and~$\capi_{e^*} = r+1$.
  It is possible to disable the inflow restriction for some arc $e = uv$ by choosing an inflow
  capacity~$\capi_e$ larger than the potential total inflow into~$u$, namely~$\sum_{e \in \delta^-(u)} \capo_{e}$.
  
  \paragraph{Flows over time.} 
  The time-depending flows we consider here are specified by $f = (f^+_e, f^-_e)_{e\in E}$,
    where~$f^+_e, f^-_e\colon [0,\infty) \rightarrow [0, \infty)$ are locally integrable and bounded functions
  for every arc~$e$. The function~$f^+_e$ describes the \emph{inflow rate} and~$f^-_e$ the \emph{outflow rate} of
  arc~$e$ for every given point in time~$\theta \in [0, \infty)$. The cumulative in- and outflow functions are defined
  as follows:
  \[F^+_e(\theta) \coloneqq \int_{0}^{\theta} f_e^+(\xi) \diff \xi \quad \text{ and } \quad
    F^-_e(\theta) \coloneqq \int_{0}^{\theta} f_e^-(\xi) \diff \xi.\]
  Due to technical reasons we define $F_e^+(\theta) = F_e^-(\theta) = 0$ for $\theta < 0$. Note that $F_e^+$ and $F_e^-$ are monotonically increasing and Lipschitz continuous. We say $f = (f^+_e,
  f^-_e)_{e \in E}$ is a \emph{flow over time} if it \emph{conserves flow} at every node~$v \in V \backslash \Set{t}$,
  i.e., if for all $\theta \in [0, \infty)$ the following equation holds
  \[\sum_{e\in \delta^+(v)} f_e^+(\theta) - \sum_{e \in \delta^-(v)} f_e^-(\theta) = \begin{cases}
  0 & \text{ if } v \not= s, \\
  r & \text{ if } v = s.
  \end{cases}\]
  \paragraph{Queues.} If more flow wants to leave~$e$ than possible, a queue builds up, which we imagine as a point queue at the head of
  the arc, as depicted in \Cref{fig:example}. The amount of flow in the queue at time~$\theta$ is given by $z_e(\theta) \coloneqq F_e^+(\theta - \tau_e) -
  F_e^-(\theta)$. Note that flow always leaves the queue as fast as possible, which is indirectly implied by the
  feasibility conditions below.
  
  \paragraph{Full arcs.} The \emph{arc load} is the total amount of flow on an arc $e$ given 
  by~$d_e(\theta) \coloneqq F_e^+(\theta) - F_e^-(\theta)$. It is the sum of the flow traversing the arc and the flow in
  the queue at a point in time~$\theta$. We say the arc is \emph{full} at time~$\theta$ if~$d_e(\theta) = \sigma_e$. For
  technical reasons we also say an arc~$e$ is full if~$d_e(\theta) > \sigma_e$ even though we show in
  \Cref{lem:additional_conditions} that this can never happen for a feasible flow over time. 
  
  \paragraph{Flow bounds.}
  The \emph{inflow bound} $\b_e^+$ is defined
 by 
  \[\b_e^+(\theta) \coloneqq \begin{cases}
  \min\set{f^-_e(\theta), \capi_e} & \text{if $e$ is full at $\theta$},\\
  \capi_e & \text{else,}
  \end{cases}\]
  and the \emph{push rate} $\b_e^-$ of arc~$e$ is given by 
  \[ \b_e^-(\theta) \coloneqq \begin{cases}
    \capo_e & \text{ if } z_e(\theta) > 0, \\
    \min\Set{f_e^+(\theta - \tau_e), \capo_e} & \text{ if } z_e(\theta) \leq 0.
    \end{cases}\]
  The value~$\b_e^-(\theta)$ describes the rate, with which the flow would leave arc~$e$ at time~$\theta$ if it is not
  restricted by any spillback. Obviously, this is an upper bound on the actual outflow rate~$f_e^-(\theta)$, which is
  captured by the fair allocation condition below. Due to spillback it is possible that for some arc $e$ we have $f_e^-(\theta) < \b_e^-(\theta)$. In this case we call $e$
  \emph{throttled} at time~$\theta$.
  
\paragraph{Feasibility.}
  A flow over time $f$ is \emph{feasible} if it satisfies the following four conditions:
  \begin{itemize}
  \item \emph{Inflow condition:} We have~$f^+_e(\theta) \leq \b_e^+(\theta)$ for all~$\theta$ and every arc~$e \in E$.
  \item \emph{Fair allocation condition:} For every node~$v$ at time~$\theta$ there is a~$\fac_v(\theta) \in (0,1]$ such
  that $f_e^-(\theta) = \min\Set{\b_e^-(\theta), \capo_e \cdot
    \fac_v(\theta)}$ for all incoming arcs~$e \in \delta^-(v)$.
  \item \emph{No slack condition:} For every node~$v$ it holds that if there is an incoming arc that is throttled at
  time~$\theta$, then there has to be at least one outgoing arc~$e \in \delta^+(v)$ with~$f^+_e(\theta) =
  \b_e^+(\theta)$. 
  \item \emph{No deadlock condition:} For every point in time~$\theta$ the set of full arcs is cycle free.
  \end{itemize}
Intuitively, the fair allocation condition ensures that the total node inflow is shared among the incoming throttled
arcs proportionally to their outflow capacities (similar to the zipper method in traffic) and the no slack conditions ensures that no arc is
throttled causeless. Since the fair allocation condition would fail in the case of a cycle of full arcs, we exclude this scenario in our model. In a Nash flow over time a deadlock can never occur anyway.

In a feasible flow over time the outflow rates never exceed the outflow capacities, queues never become negative, and the arcs don't get overfull:
\begin{lemma} \label{lem:additional_conditions}
  A feasible flow over time $f$ satisfies the following conditions for all $\theta$ and every arc~$e$:
  \begin{enumerate}
    \item Outflow capacity condition: $f^-_e(\theta) \leq \capo_e$. \label{it:capo_condition}
    \item Non-deficit condition: $z_e(\theta) \geq 0$. \label{it:no_deficit_condition}
    \item Storage condition: $d_e(\theta) \leq \sigma_e$. \label{it:storage_condition}
  \end{enumerate}
\end{lemma}
\begin{proof}\phantom{a}
    \begin{enumerate}
      \item[\ref{it:capo_condition}] This follows immediately from the fair allocation condition.       
      \item[\ref{it:no_deficit_condition}] Assume for contradiction that $z_e(\theta)<0$ at some point. Since $z_e$ is continuous we find an interval $(\theta_0, \theta_1]$ with $z_e(\theta_0) = 0$ and $z_e(\theta) < 0$ for all $\theta \in (\theta_0, \theta_1]$.
      By the fair allocation condition and the definition of the push rate for the case of $z_e(\theta) \leq 0$ it follows that $f_e^-(\theta) \leq
      \b_e^-(\theta) \leq f_e^+(\theta - \tau_e)$ for all $\theta \in [\theta_0, \theta_1]$. This leads to a
      contradiction:
      \[0 > z_e(\theta_1) - z_e(\theta_0) = \int_{\theta_0}^{\theta_1} f_e^+(\xi - \tau_e) - f_e^-(\xi) \diff \xi
      \geq 0.\]
      \item[\ref{it:storage_condition}] Assume for contradiction that $d_e(\theta)>\sigma_e$ at some point. Since $d_e$ is continuous we find an interval $(\theta_0, \theta_1]$ with $d_e(\theta_0) = \sigma_e$ and $d_e(\theta) > \sigma_e$ for all $\theta \in (\theta_0, \theta_1]$. From the inflow condition it follows that
      $f_e^+(\theta)\leq f_e^-(\theta)$ for all $\theta \in [\theta_0, \theta_1]$. Again this leads to a contradiction, since we have
      \[0 < d_e(\theta_1) - d_e(\theta_0) = \int_{\theta_0}^{\theta_1} f_e^+(\xi) - f_e^-(\xi) \diff \xi 
      \leq 0.\]      
    \end{enumerate}
    \end{proof}

\begin{lemma}\label{lem:arc_saturation}
  If $e$ is full at time $\theta$, we have $z_e(\theta) > 0$.
\end{lemma}
\begin{proof}
Using the inflow condition $f_e^+(\theta) \leq \capi_e$ and the requirement that $\sigma_e > \capi_e \cdot \tau_e$ we obtain
\[z_e(\theta) = F_e^+(\theta - \tau_e) - F_e^-(\theta) \geq F_e^+(\theta) - \capi_e \cdot \tau_e - F_e^-(\theta) = \sigma_e - \capi_e \cdot \tau_e >0.\]
\end{proof}

\paragraph{Spillback factor.}
For every node $v$ we call the \emph{maximal} value $\fac_v(\theta) \in (0,1]$ that fulfills the fair allocation condition
the \emph{spillback factor} for node $v$ at time~$\theta$.

  If more flow wants to enter a node than the outgoing arcs can handle, the spillback factor will be strictly less than $1$.
  In this case the fair allocation condition ensures that flow conservation holds by temporarily reducing the outflow capacity of incoming arcs to $c_v(\theta) \cdot \capo_e$.
  \paragraph{Travel and arrival times.} Given a network and a feasible flow over time, an important question is at which
  time a sample flow particle starting at~$s$ at time~$\theta$ can reach a node~$v \in V$. First, we consider the \emph{waiting time} in the queue for a particle entering an arc~$e$ at time~$\theta$, which is given by 
  \[q_e(\theta) \coloneqq \min \Set{ q  \geq 0| \int_{\theta + \tau_e}^{\theta  + \tau_e + q} f_e^-(\xi) \; \diff \xi 
  = z_e(\theta  + \tau_e)}.\] 
  To show that the set on the right hand side is never empty, and thus $q_e(\theta)$ is well-defined, we prove that
  there is a network-wide lower bound on the outflow rate of arcs with positive queues.  

  \begin{restatable}{lemma}{qwelldefined}
  \label{lem:q_well-defined}
    For a given network there is an $\epsilon>0$ such that for every arc~$e$ with~$z_e(\theta)>0$ we have
    $f_e^-(\theta)\geq\epsilon$, and therefore, the waiting time function $q_e$ is well-defined.
  \end{restatable}
  
    \begin{proof}
    There can be a long chain of full arcs behind $e$ reducing the outflow rate significantly. But due to the no
    deadlock and the no slack condition there has to be an arc where the inflow capacity is
    exhausted. Using the fair allocation condition it is possible to choose $\epsilon$ only depending on the smallest
    capacity and the total number of arcs. 
    
   We set $\epsilon \coloneqq \left(\frac{\capmin}{\Sigma}\right)^{\abs{E}} \cdot \capmin$, where $\capmin\coloneqq \min \left(\set{ \capi_e, \capo_e | e \in E } \cup \set{1}\right)$ and $\Sigma \coloneqq \max\Set{\sum_{e \in E} \capo_e, 1}$. 
  If $e$ is not throttled, we have $f_e^-(\theta)=\b_e^-(\theta)=\capo_e \geq \capmin \geq \epsilon$. 
  So suppose $e$ is throttled. By the no slack condition there has to be a consecutive arc $e_1$ with $f_{e_1}^+(\theta) = b_{e_1}^+(\theta)$.
  If $e_1$ is full and throttled, we consider the next arc $e_2$, where the inflow bound is exhausted. We continue until we find an arc $e_k$ that is not full or not throttled. Since the set of full arcs is cycle free by the no deadlock condition this sequence $(e_1 , \ldots , e_k)$ is finite with $k\leq m$. By \Cref{lem:arc_saturation} we have that $f_{e_k}^+(\theta) = \b_e^+(\theta) \geq \min\set{\capi_{e_k}, \capo_{e_k}}$.
  Furthermore, for two consecutive arcs $e_{i-1}=uv$ and $e_i=vw$ we have
  \begin{equation} \label{eq:lowerbound_on_outflow} f_{e_{i-1}}^-(\theta) = \fac_v(\theta) \cdot \capo_{e_{i-1}} \geq \frac{\sum_{e' \in \delta^+(v)}f_{e'}^+(\theta)}{\sum_{e' \in \delta^-(v)}\capo_{e'}} \cdot \capmin \geq \frac{f_{e_i}^+(\theta)}{\Sigma}\cdot \capmin. \end{equation}
  Since the arcs $e_1, \dots, e_{k-1}$ are full with exhausted inflow capacity it holds that $f_{e_i}^+(\theta)=\b^+_{e_i}(\theta) = \min\set{f_{e_i}^-(\theta), \capo_{e_i}}$. Recursive application of
   \eqref{eq:lowerbound_on_outflow} along the sequence gives $f_e^-(\theta)\geq \left(\frac{\capmin}{\Sigma}\right)^k \cdot \capmin \geq \epsilon$.
  
  Next, we show that the set in the definition of $q_e(\theta)$ is not empty. In the case that there exists a $\theta'  \geq \theta$ such that $z_e(\theta ' + \tau_e)=0$ we have
  \[0=z_e(\theta ' + \tau_e) \geq F_e^+(\theta) - F_e^-(\theta ' + \tau_e) 
  = z_e(\theta+ \tau_e) - \int_{\theta + \tau_e}^{\theta ' + \tau_e} f_e^-(\xi) \; \diff \xi.\]
  Thus, there exists a $q \in [0, \theta ' - \theta]$ that is in the set.
  In the case that $z_e(\theta' + \tau_e)>0$, and thus $f_e^-(\theta' + \tau_e)\geq \epsilon$, for all $\theta' \geq
  \theta$ we have that $\int_{\theta + \tau_e}^{\theta + \tau_e + q} f_e^-(\xi) \; \diff \xi \to \infty$ for $q \rightarrow \infty$. Since $z_e(\theta + \tau_e)$ is a fixed value there has to be a $q$ that satisfies the condition of the set.
  Hence, the set is non-empty and due to the continuity in $q$ it is closed, which shows that the minimum exists.
\end{proof}

  
  A particle entering an arc~$e$ at time~$\theta$ first traverses the arc in $\tau_e$ time, then waits in the queue for
  $q_e(\theta)$ time units before it leaves the arc at the \emph{exit time} $T_e(\theta)\coloneqq \theta + \tau_e + q_e(\theta)$.
  We denote the time a particle starting at time~$\theta$ needs to traverse a path~$P=(e_1, \dots, e_k)$ by $T_P(\theta)\coloneqq T_{e_k} \circ \ldots \circ T_{e_1}(\theta)$.
  The \emph{earliest arrival time function} $\l_v\colon [0, \infty) \rightarrow [0, \infty)$ maps a time $\theta$ to the
  earliest time a sample particle, starting at $\theta$ at $s$, can reach $v$, i.e.,
  $\l_v(\theta) \coloneqq \min_{P \in \paths_v} T_P(\theta)$,
  where $\paths_v$ denotes the set of all $s$-$v$-paths. 
  They are also characterized by the
  following dynamic Bellman's equations:
 \begin{equation}\label{eqn:bellman}\l_v(\theta) = \begin{cases}
 \qquad\theta & \text{ if } v =s, \\
 \min\limits_{e=uv \in \delta^-(v)} T_e(\l_u(\theta))& \text{ if } v \neq s,
 \end{cases} \quad \text{ for all } v \in V.\end{equation}
Since we require the transit times of all directed cycles to be positive this is well defined.
         
The following lemma is a collection of technical properties, which will be useful later on.
  \begin{restatable}{lemma}{technicalproperties} \label{lem:technical_properties}
    For a feasible flow over time $f$ it holds for all $e \in E$, $v \in V$ and~$\theta \geq 0$ that:
    \begin{enumerate}
      \item $F^+_e(\theta)=F^-_e(T_e(\theta))$. \label{it:fifo}
      \item $q_e(\theta) > 0 \;\; \Leftrightarrow \;\; z_e(\theta+\tau_e) > 0$. \label{it:q_equiv_z}
      \item For $\theta_1 < \theta_2$ with $F_e^+(\theta_2) - F_e^+(\theta_1) = 0$, and $z_e(\theta_2 + \tau_e)>0$  
      we have $T_e(\theta_1)=T_e(\theta_2)$. \label{it:equal_exit_times}
      \item If $f_e^-(T_e(\theta)) = 0$ then $F_e^+(\theta + q_e(\theta)) - F_e^+(\theta) = 0$. \label{it:no_inflow}
      \item For the push rate function it holds that
      \[\b_e^-(T_e(\theta)) = \begin{cases}
      \capo_e &\text{ if }F_e^+(\theta + q_e(\theta)) - F_e^+(\theta) > 0,\\
      \min\set{f_e^+(T_e(\theta) - \tau_e), \capo_e} &\text{ else.}
      \end{cases}\!\] \label{it:maxoutflow_depending_on_inflow}
      \item We have $z_e(\theta +\tau_e + \xi) > 0$ for all $\xi \in [0, q_e(\theta))$.
      \label{it:positive_queue_while_emptying}
      \item The function $T_e$  and $\l_v$ are monotonically increasing. \label{it:T_monoton}
      \item The functions $q_e$, $\ell_v$ and $T_e$ are Lipschitz continuous. \label{it:q_is_lipschitz}
    \end{enumerate}
  \end{restatable}
The technical proof can be found in the appendix.
 Note that the Lipschitz continuous functions $q$, $T$, and $\l$ are almost everywhere
  differentiable due to Rademacher's theorem; see~\cite{rademacher1919partielle}.   
  The derivatives of the waiting times are described in the following lemma.
\begin{lemma} \label{lem:q'}
For almost all $\theta$ the following is true:
  \[q'_e(\theta) = \begin{cases}
  \frac{f^+_e(\theta)}{f^-_e(T_e(\theta))} - 1 &\text{if } f^-_e(T_e(\theta)) > 0, \\
  -1 &\text{else if } z_e(\theta + \tau_e) > 0,\\
  0 & \text{else.}
  \end{cases}\]
\end{lemma}
\begin{proof}
  By definition of $q_e(\theta)$ we have
  $F_e^-(T_e(\theta)) - F_e^-(\theta+\tau_e) 
  = z_e(\theta + \tau_e)$.
  Since the functions $F_e^-$, $z_e$, and $q_e$ are almost everywhere differentiable we can
  take the derivative on both sides to obtain
  \[f_e^-(T_e(\theta)) \cdot (1 + q'_e(\theta)) - f_e^-(\theta+\tau_e) = z'_e(\theta + \tau_e).\]
  Since $z_e'(\theta + \tau_e) =
  f_e^+(\theta) - f_e^-(\theta + \tau_e)$ we get that
  $q'_e(\theta) = \frac{f^+_e(\theta)}{f^-_e(T_e(\theta))} - 1$ if $f_e^-(T_e(\theta)) > 0$.
  
  In the case of $f_e^-(T_e(\theta))=0$ and $z_e(\theta + \tau_e) > 0$ \Cref{lem:technical_properties}\ref{it:no_inflow} yields $F_e^+(\theta + \xi) - F_e^+(\theta) = 0$ for all $\xi \in [0, q_e(\theta)) \not=
  \emptyset$, and therefore, $T_e(\theta) = T_e(\theta + \xi)$ by \Cref{lem:technical_properties}~\ref{it:q_equiv_z} and~\ref{it:equal_exit_times}.
  It follows that  
 \[q_e(\theta + \xi) = T_e(\theta + \xi) - \theta - \xi - \tau_e
  = T_e(\theta) - \theta - \tau_e - \xi = 
  q_e(\theta) - \xi.\]
  Hence, the right derivative of $q_e$ at~$\theta$ equals $-1$. Hence, either $q$ is not differentiable at~$\theta$ or $q'_e(\theta) = -1$.
  
  Finally, in the case of $f_e^-(T_e(\theta))=0$ and $z_e(\theta + \tau_e) = 0$ we have
  $q_e(\theta) = 0$ by \Cref{lem:technical_properties}~\ref{it:q_equiv_z}, and thus, $\theta$ is a local minimum of $q_e$. Hence, $q_e$ is either not differentiable at $\theta$ or $q_e'(\theta )=0$.
\end{proof}

  \paragraph{Active, resetting and spillback arcs.} For every point in time~$\theta$ we define the following classes of
  arcs. We say an arc is \emph{active} for
  $\theta$ if it attains the minimum in \eqref{eqn:bellman}, i.e., the set of active arcs is \[E_{\theta}'=\Set{e=uv \in E  | \l_v(\theta)=T_e(\l_u(\theta))}.\]
  The subgraph $G_\theta' \coloneqq (V, E'_\theta)$, is called \emph{current shortest paths network}. Note that
  this graph is acyclic since all directed cycles have a positive transit time.
  Furthermore, every node is reachable from $s$ in $G_\theta'$ since the indegree of all those nodes is positive.
      
  We call the set of arcs on which the particle entering at time~$\theta$ would experience a queue \emph{resetting arcs}
  and arcs that are full when the particle would arrive there are called \emph{spillback arcs}. We denote them by
  \[E_\theta^* \coloneqq \Set{\!e \!=\! uv \in E | q_e(\l_u(\theta)) > 0\!} \text{ and } \bar E_\theta \coloneqq \Set{\!e \!=\! uv \in E | d_e(\l_u(\theta)) = \sigma_e\!}.\]
  
  \section{Nash Flows Over Time and Spillback Thin Flows} \label{sec:Nash_flows}
  In this section we define a dynamic
  equilibrium, called \emph{Nash flow over time}, for the spillback model and we show, as a central structural result,
  that the strategy of every particle can be described by a particular static flow, which we call \emph{spillback thin
  flow}.

  \paragraph{Nash flows over time.} A feasible flow over time $f = (f^+_e, f^-_e)_{e\in E}$ is a \emph{Nash flow over
  time} if it satisfies the \emph{Nash flow condition}, i.e., for almost all $\theta \in [0, \infty)$ and all arcs $e =
  uv$ we have
    \begin{equation*}
    f_e^+(\theta) > 0 \quad \Rightarrow \quad \theta \in \Set{\l_u(\vartheta) \in [0, \infty) | e \in E'_\vartheta}.
    \end{equation*} 
  
  \begin{remark}
  A game theoretical Nash equilibrium is a state such that no player can improve by choosing an alternative strategy.
  Since for every particle starting in~$s$ at time~$\theta$ the earliest possible arrival time~$\l_t(\theta)$ is
  realized, there is no improving move from the perspective of a single particle.
   \end{remark}

The following lemma give some chracterizations of a Nash flow over time.
\begin{restatable}{lemma}{xwelldefined}\label{lem:x_well-defined}
Let $f$ be a feasible flow over time. The following statements are equivalent.
\begin{enumerate}
  \item $f$ is a Nash flow over time. \label{it:nash_flow}
  \item $F^+_e(\l_u(\theta)) = F_e^-(\l_v(\theta))$ for all arcs $e = uv$ and all times~$\theta$. \label{it:in_equals_out_at_l}
  \item For all $\theta \geq 0$ and all $e = uv \in E$: 
  If $F^+_e(\l_u(\theta) - \epsilon) < F^+_e(\l_u(\theta))$
   for all $\epsilon > 0$ then $e \in E'_\theta$. \label{it:F_increasing_means_active}
\end{enumerate}
\end{restatable}
The proof of $\ref{it:nash_flow}\Leftrightarrow \ref{it:in_equals_out_at_l}$ can be found in \cite[Theorem 1]{cominetti2015dynamic}. The equivalence to \ref{it:F_increasing_means_active} is shown in the appendix.

The active, resetting, and spillback arcs in a Nash flow over time have the following properties:  
  \begin{lemma} \label{lem:resetting_implies_active}
  Given a Nash flow over time the following holds for all times $\theta$:
  \begin{enumerate}
    \item $E^*_\theta \subseteq E'_\theta$. \label{it:resetting_subset_active}
    \item $E'_\theta = \set{e = uv | \l_v(\theta) \geq \l_u(\theta) + \tau_e}$. \label{it:characterization_of_E'}
    \item $E^*_\theta = \set{e = uv | \l_v(\theta) > \l_u(\theta) + \tau_e}$. \label{it:characterization_of_E^*}
    \item $\bar E_\theta \subseteq E'_\theta$. \label{it:full_subset_active}
    \item $\l_u(\theta) < \l_v(\theta)$ for all $e = uv \in \bar E_\theta$. \label{it:full_arc_need_time} 
  \end{enumerate}  
  \end{lemma}
\begin{proof}
  For the proof of \ref{it:resetting_subset_active} to \ref{it:characterization_of_E^*} see \cite[Proposition 2]{cominetti2015dynamic}.  
  \begin{enumerate}
 \item[\ref{it:full_subset_active}] Since $e$ is full at time $\l_u(\theta)$ we have by \Cref{lem:arc_saturation} that
 $z_e(\l_u(\theta)) >0$.
 Therefore, by continuity of $z_e$ and \Cref{lem:q_well-defined} we have that $f_e^-(\xi) > 0$ for all $\xi \in
 [\l_u(\theta) - \delta, \l_u(\theta)]$ for a small $\delta > 0$. It follows that for all $\epsilon > 0$ we have
 $F_e^-(\l_u(\theta)) - F_e^-(\l_u(\theta) - \epsilon) > 0$.
 This together with the storage condition in \Cref{lem:additional_conditions}~\ref{it:storage_condition} yields
 \begin{align*}
 F_e^+(\l_u(\theta)) - &F_e^+(\l_u(\theta) - \epsilon)\\
 &= d_e(\l_u(\theta)) + F_e^-(\l_u(\theta)) - d_e(\l_u(\theta)-\epsilon) - F_e^-(\l_u(\theta)-\epsilon) \\
 &> \sigma_e - d_e(\l_u(\theta)-\epsilon) \geq 0.
 \end{align*} 
 Hence, \Cref{lem:x_well-defined}~\ref{it:F_increasing_means_active} implies $e \in E'_{\theta}$.
  \item[\ref{it:full_arc_need_time}] 
  
  Due to \ref{it:full_subset_active}, $e$ is active, i.e., $ \l_u(\theta) + \tau_e + q_e(\l_u(\theta)) = \l_v (\theta)$. Thus, $\tau_e > 0$ is clear. If $\tau_e = 0$
we get by \Cref{lem:arc_saturation} that $0 < z_e(\l_u(\theta)) = z_e(\l_u(\theta)+\tau_e)$. Hence, \Cref{lem:technical_properties}~\ref{it:q_equiv_z} implies
  $q_e(\l_u(\theta)) > 0$.   
%
 \end{enumerate}
\end{proof}

It is worth noting that in general we do not have $\bar E_\theta \subseteq E^*_\theta$.
An arc which is full at time $\l_u(\theta)$ always has a positive queue at this point in time. Though it is possible that the queue depletes until time $\l_u(\theta)+\tau_e$, in which case $e \not\in E^*_\theta$.
  
  \paragraph{Underlying static flows.}
  \Cref{lem:x_well-defined}~\ref{it:in_equals_out_at_l} motivates to define the \emph{underlying static flows} for all $\theta$:
  \[x_e(\theta) \coloneqq F^+_e(\l_u(\theta)) = F_e^-(\l_v(\theta)).\]
 It is easy to verify that for a fixed time~$\theta$ this is indeed a static $s$-$t$-flow of flow value~$r \cdot \theta$ and
 that $x_e$ as a function is monotonically increasing and Lipschitz continuous. 
 Applying Rademacher's theorem to $x_e$
 and $\l_v$ we obtain derivatives~$x'_e(\theta)$ and $\l'_v(\theta)$ almost everywhere.
It is possible to reconstruct the Nash flow over time by these derivative functions, since 
\[x'_e(\theta) = f_e^+(\l_u(\theta)) \cdot \l'_u(\theta) = f_e^-(\l_v(\theta)) \cdot \l'_v(\theta).\]
Furthermore, $x'(\theta)$ forms a
static $s$-$t$-flow of value~$r$ and can be seen as the strategy of the flow entering the network at time~$\theta$. In
other words, these derivative functions characterize a Nash flow over time and it turns out that they have a very particular
structure, which we call \emph{spillback thin flows}. This is a generalization of \emph{thin flows with resetting} introduced by~\cite{koch2010nash}, as explained in \Cref{prop:generalization} below.
  
  \paragraph{Spillback thin flows.}
  Consider an acyclic directed graph~$G' = (V, E')$ with a source~$s$ and a sink~$t$ where all nodes are
  reachable from~$s$. Every arc~$e$ is equipped with an outflow capacity~$\capo_e>0$ and an inflow bound~$\b_e^+>0$.
  Additionally, we are given a subset of arcs~$E^*\subseteq E'$. 
  A static $s$-$t$-flow $x'$ of value $r$ (which does not need to obey the capacities) together with two node
  labelings $\l'_v \geq 0$ and $\fac_v \in (0,1]$ is a \emph{spillback thin flow} with resetting on $E^*$ if it
  fulfills the following equations:
  \begin{alignat}{2}
  \l'_s &= \frac{1}{\fac_s}, \label{eqn:l'_s}\tag{TF1}\\  
  \l'_v &= \min_{e = uv \in E'} \rho_e(\l'_u, x'_e,\fac_v) \quad 
   \text{ for } v \in V\backslash\Set{s},\label{eqn:l'_v_min}\tag{TF2}\\
  \l'_v &= \rho_e(\l'_u, x'_e, \fac_v)
   \text{ for } e=uv \in E' \text{ with } x'_e > 0, \label{eqn:l'_v_tight}\tag{TF3}\\
  \l'_v &\geq \max_{e = vw \in E'} \frac{x'_e}{\b_e^+} \;\;
   \text{ for } v \in V, \label{eqn:l'_v_min_blocked}\tag{TF4}\\
  \l'_v &= \max_{e = vw \in E'} \frac{x'_e}{\b_e^+} \;\;
   \text{ for } v \in V \text{ with } \fac_v < 1, \label{eqn:l'_v_min_blocked_equal}\tag{TF5}
  \end{alignat}
  where \[\rho_e(\l'_u, x'_e, \fac_v) \coloneqq \begin{cases}
  \frac{x'_e}{\fac_v \cdot \capo_e}& \text{if } e = uv \in E^*,\\
  \max\Set{\l'_u, \frac{x'_e}{\fac_v \cdot \capo_e}}\!& \text{if } e = uv \in E'\backslash E^*.
  \end{cases}\]
  The next theorem describes the relation between spillback thin flows and Nash flows over time.
  \begin{theorem}
  \label{thm:nash_flow_derivatives_are_thin_flow}
    For almost all~$\theta \in [0, \infty)$ the derivatives $x'_e(\theta)$ and $\l'_v(\theta)$ of a Nash flow over time together with the spillback factors
    $\fac_v(\l_v(\theta))$ form  a spillback thin flow on the current shortest
    paths network $G'_\theta = (V, E'_\theta)$ with resetting on the arcs with queue $E_\theta^*$ and inflow bounds
    $\b_e^+(\l_u(\theta))$.
  \end{theorem} 
  \begin{proof}
  We fix a point in time $\theta$ such that for all $e = uv \in E$ the derivatives of $x_e$, $\l_v$, and $T_e \circ \l_u$ exist and $x_e'(\theta) = f_e^-(\l_v(\theta)) \cdot \l'_v(\theta) = f_e^+(\l_u(\theta)) \cdot \l'_u(\theta)$. Note that almost all points in time satisfy these conditions.
  For short, let $\l_v' \coloneqq \l_v'(\theta)$, $x'_e \coloneqq x'_e(\theta)$, $\fac_v \coloneqq \fac_v(\l_v(\theta))$, $b_e^+
  \coloneqq b_e^+(\l_u(\theta))$, $E' \coloneqq E'_\theta$, and $E^* \coloneqq E^*_\theta$.
%

 \paragraph{(\ref{eqn:l'_s})} We have $\l_s(\theta) = \theta$ yielding $\l'_s=1$. 
 By assumption $\delta^-(s)=\emptyset$, as well as, $\sigma_e = \infty$ and $\capi_e > r$ for $e \in \delta^+(s)$. 
 Hence, the no slack condition implies $\fac_s = 1$.

  \paragraph{(\ref{eqn:l'_v_min})}  
  
    By differentiating $\l_v(\theta) = \min_{e = uv \in E} T_e(\l_u(\theta))$,
  we obtain that
  \[\l'_v = \min_{e = uv \in E'} T'_e(\l_u(\theta)) \cdot \l'_u.\]
  Note that $E'$ is exactly the set of arcs with $\l_v(\theta) = T_e(\l_u(\theta))$, and therefore, exactly these need to be considered for the derivative.
  In the following we analyze the derivative of $T_e(\theta) = \theta + \tau_e + q_e(\theta)$ at the point $\l_u(\theta)$
  for active arcs $e = uv \in E'$. \Cref{lem:q'} yields  
  \[T'_e(\l_u(\theta)) = \begin{cases}\frac{f_e^+(\l_u(\theta))}{f_e^-(\l_v(\theta))} &\text{ if } f_e^-(\l_v(\theta)) > 0, \\
  0 &\text{ else if } z_e(\l_u(\theta) + \tau_e) > 0,\\
  1 & \text{ else.}
  \end{cases}\]
  First, we consider the case $f_e^-(\l_v(\theta)) = 0$, which implies $x'_e = 0$, and hence,
  \begin{align*} T'_e(\l_u(\theta)) \cdot \l'_u = \left\{\begin{array}{ll}  
  0 &\quad\text{if } q_e(\l_u(\theta)) > 0,\\
   \l'_u &\quad\text{else},
  \end{array}\right\} 
  = \rho_e(\l'_u, x'_e, \fac_v).
  \end{align*}
  Next, we consider the case $f_e^-(\l_v(\theta)) > 0$ and $x'_e = 0$. If $e \not \in E^*$, we  have
  $f_e^+(\l_u(\theta)) = f_e^+(\l_v(\theta) - \tau_e) \geq \b_e^-(\l_v(\theta)) \geq f_e^-(\l_v(\theta)) > 0$, which
  implies $\l'_u = \frac{x_e'}{f_e^+(\l_u(\theta))} = 0$. In both cases, whether $e \in E^*$ or not, we have
  $T'_e(\l_u(\theta)) \cdot \l'_u = \frac{x'_e}{f_e^-(\l_v(\theta))} = 0 = \rho_e(\l'_u, x'_e, \fac_v)$.
  
  Finally, we consider $f_e^-(\l_v(\theta)) > 0$ and $x'_e > 0$. This implies that $x_e(\theta) = F_e^+(\l_u(\theta))$ is strictly increasing in $[\l_u(\theta), \l_u(\theta) + \epsilon]$, and therefore,
  $F_e^+(\l_u(\theta) + q_e(\l_u(\theta))) - F_e^+(\l_u(\theta)) > 0$ if and only if $q_e(\l_u(\theta)) > 0$.
  We obtain together with \Cref{lem:technical_properties}~\ref{it:maxoutflow_depending_on_inflow} that
  \[\b_e^-(\l_v(\theta)) = \begin{cases}
  \capo_e & \text{ if } e \in E^*,\\
  \min\set{f_e^+(\l_u(\theta)), \capo_e} & \text{ if }e \in E' \backslash E^*.
  \end{cases}\]
  Hence,
  \begin{equation} \label{eqn:T'l'_equals_rho}
  \begin{aligned}
  T'_e(\l_u(\theta)) \cdot \l'_u &= \frac{x'_e}{f_e^-(\l_v(\theta))}
  = \frac{x'_e}{\min\set{\fac_v \cdot \capo_e, \b_e^-(\l_v(\theta))}}\\
  &= \begin{cases}  \frac{x'_e}{\fac_v \cdot \capo_e} & \text{ if } e \in E^*\\
  \max\Set{\frac{x'_e}{f_e^+(\l_u(\theta))}, \frac{x'_e}{\fac_v \cdot \capo_e}} & \text{ if } e \in E'\backslash E^*
  \end{cases}\\
  &= \rho_e(\l'_u, x'_e, \fac_v).
  \end{aligned}
  \end{equation}
In summary, we have
\[\l'_v = \min_{e = uv \in E'} T'_e(\l_u(\theta)) \cdot \l'_u = \min_{e = uv \in E'} \rho_e(\l'_u, x'_e, \fac_v).\]
  
  \paragraph{(\ref{eqn:l'_v_tight})} Suppose $x'_e = f_e^-(\l_v(\theta)) \cdot \l'_v  > 0$. With
  \eqref{eqn:T'l'_equals_rho} we get $\l'_v = \frac{x'_e}{f_e^-(\l_v(\theta))} = \rho_e(\l'_u, x'_e, \fac_v)$.

  \paragraph{(\ref{eqn:l'_v_min_blocked})}
  By the inflow condition we have for all arcs $e=vw$ that $x'_e = f_e^+(\l_v(\theta)) \cdot \l'_v \leq \b_e^+ \cdot \l'_v$.
  
  \paragraph{(\ref{eqn:l'_v_min_blocked_equal})} Suppose we have $\fac_v < 1$. The maximality of $\fac_v$ implies that there has to be at least one incoming throttled arc and by the no slack condition there has to be an outgoing arc $e = vw$ with $f^+_e(\l_v(\theta))=\b_e^+$. Hence,
  $x'_e = f_e^+(\l_v(\theta)) \cdot \l'_v = \b_e^+ \cdot \l'_v$. Together with (\ref{eqn:l'_v_min_blocked}) we obtain (\ref{eqn:l'_v_min_blocked_equal}).
 \end{proof}

  \section{Computation of Nash Flows Over Time with Spillback} \label{sec:constructing_nash_flows} In this section we
  show how to construct a Nash flow over time with spillback for a given network using spillback thin flows. The key
  idea is to start with the empty flow over time and to extend it step by step. For this we first show that for all
  acyclic networks $G'= (V, E')$ with arbitrary capacities, outflow bounds, and resetting arcs $E^*$ there always exists
  a spillback thin flow.
  \begin{theorem}
  \label{thm:existence_of_thin_flow}
    Given an acyclic network~$G'= (V, E')$ with source~$s$ and sink~$t$, such that each node is reachable from~$s$, let $(\capo_e)_{e\in E'}$ be outflow capacities, $(\b_e^+)_{e\in E'}$ be inflow bounds, and
    $E^*\subseteq E'$ be a set of arcs. Then there exists a spillback thin flow~$(x', \l', \fac)$ with resetting on~$E^*$.
  \end{theorem}
  The proof uses an existence result for variational inequalities.
Let $I$ be a finite index set, $K \subseteq \R^I$ and $\Gamma\colon K \rightarrow \R^I$. The \emph{variational inequality problem} $\VI(K, \Gamma)$ is to find a vector $x \in K$ such that

\begin{equation}(y - x)^t\Gamma(x) \geq 0, \quad \forall y \in K. \tag{VI}\label{eqn:VI}
\end{equation}
    
   The set of solution set $\SOL(K, \Gamma)$ is non-empty, which can be seen by Brouwer's fixed point theorem; see \cite{harker1990viexistence} for details.
   
   \begin{theorem}[{\cite[Theorem 3.1]{harker1990viexistence}}] \label{thm:solution_of_VI}
   Let $K \subseteq \R^I$ be non-empty, compact and convex and let $\Gamma\colon K \rightarrow \R^I$ be a continuous mapping. Then $\SOL(K,\Gamma)$ is non-empty.
   \end{theorem}
     
If $K$ is a box, i.e., $K = \bigtimes_{i \in I}[0, M_i]$ for some $M_i > 0$, it is easy to see that for a given solution $x^* \in \SOL(K, \Gamma)$ the \emph{nonlinear complematary problem} holds for every $i \in I$ with $x_i^* < M_i$: 
     \begin{equation}\Gamma_i(x^*) \geq 0 \qquad \text{ and } \qquad x^*_i \cdot \Gamma_i(x^*) = 0.\tag{NCP} \label{eqn:NCP}
     \end{equation}
   In order to define $K$ and $\Gamma$ for our purposes let $\Vz \coloneqq \Set{\vz | v \in V}$ be a copy
   of the set of nodes~$V$ and let $I \coloneqq E' \mathbin{\dot{\cup}} \Ve \mathbin{\dot{\cup}} \Vz$ be the index set. We will see that $e \in E'$ correspond to $x'_e$, $v \in V$ to $\l'_v$ and $\vz \in \Vz$ to $\beta_v$, which corresponds
   bijectively to $\fac_v$.
    With $\capo_{\min} \coloneqq \min_{e\in E'} \capo_e$, $\capo_{\max} \coloneqq \max_{e\in E'}
   \capo_e$, and $\b_{\min}^+ = \min_{e \in E'} \b_e^+$ we define
   
   {\begin{align} \label{eqn:defi_of_M}
   M &\coloneqq \max\Set{ 1,\frac{r}{\capo_{\min}}, \frac{r}{\b_{\min}}, \frac{\capo_{\max} \cdot \abs{E'}}{\b_{\min}^+}},\\[0.5 \baselineskip]
   K &\coloneqq \Set{(x', \l', \beta) \in \R^I | \begin{array}{ll} 
   0 \leq x'_e \leq 4M^2 \cdot \capo_e & \text{ for all } e \in E'\\
   0 \leq \l'_v \leq 3M^2 & \text{ for all } v \in \Ve\\
   0 \leq \beta_v \leq \log(2M) & \text{ for all } \vz \in \Vz
   \end{array}},\notag\\[0.5 \baselineskip]
   \Gamma_i(x', \l', \beta) &\coloneqq \begin{dcases}
   \frac{x'_e}{\capo_e \cdot e^{- \beta_v}} - \l'_v  & \text{if } i = e = uv \in E^*,\\
   \max\Set{\l'_u, \frac{x'_e}{\capo_e \cdot e^{- \beta_v}}} - \l'_v  & \text{if } i = e = uv \in E'\backslash E^*,\\
   \sum_{e \in \delta^-(v)} x'_e - \sum_{e \in \delta^+(v)} x'_e & \text{if } i = v \in \Ve\backslash\set{s,t},\\
   \sum_{e \in \delta^-(t)} x'_e - \sum_{e \in \delta^+(t)} x'_e - r & \text{if } i = t \in \Ve,\\
   \l'_{s} - \frac{1}{e^{-\beta_s}}  & \text{if } i = s \in \Ve,\\
   \l'_{v} - \max\limits_{e = v w \in E'} \frac{x'_e}{\b_e^+} & \text{if } i = \vz \in \Vz.
   \end{dcases}\notag
   \end{align}}
   Since $K$ is convex and compact and $\Gamma$ is continuous there exists a
   solution $(x', \l', \beta) \in \SOL(K, \Gamma)$.
     
   \begin{lemma}\label{lem:sol_is_in_the_interior}
   For every solution $(x', \l', \beta) \in \SOL(K, \Gamma)$ we have 
   \begin{enumerate}
   \item $x'_e < 4M^2\cdot \capo_e$ for every arc $e$ \label{it:NCPx},
   \item $\l'_{v} < 3M^2$ for every node $v \in V$,
   \item $\beta_{v} < \log(2M)$ for every node $v \in V \backslash \Set{s}$ with $\sum_{e \in \delta^+(v)} x'_e > 0$.
   \end{enumerate}
   \end{lemma}
   \begin{proof}
   \begin{enumerate}
   \item Suppose there is an arc $e \in E'$ with $x'_e = 4M^2 \cdot \capo_e$. Note that $e^{-\beta_v} \leq 1$ and $\l'_{v} \leq 3M^2$, and therefore, $\Gamma_e(x', \l', \beta) =  \frac{x'_e}{\capo_e \cdot e^{-\beta_v}} - \l'_v$ even if $e \in E'\backslash E^*$. Hence, for $(y, \l', \beta) \in K$ with $y_e \coloneqq 0$, $y_i \coloneqq x'_i$ for $i \in E'\backslash\set{e}$, \eqref{eqn:VI} states that
     $0 \leq -x'_e \cdot \left(\frac{x'_e}{\capo_e \cdot e^{-\beta_v}}-\l'_v\right) 
     \leq 4M^2 \cdot \capo_e \cdot (\l'_v - 4M^2)$.
     But this is a contradiction since $\l'_v - 4M^2 < 0$. 

    \item Using $(x', k, \beta)$ with $k_v = \l'_v$ for $v \neq s$ we obtain with \eqref{eqn:VI} that $(k_s - \l'_s)
    \cdot (\l'_s - \frac{1}{e^{-\beta_s}})  \geq 0$ for all $k_s \in [0, 3M^2]$. Hence, $\l'_s = \frac{1}{e^{-\beta_s}} \leq 2 M < 3M^2$.
   We show that 
   $\sum_{e \in \delta^-(v)}x_e' \leq \sum_{e \in \delta^+(v)}x_e'$ for all $v \in \Ve \backslash\Set{ s, t}$.
   If $\l_v'>0$ this follows from \eqref{eqn:VI} for $(x',k,\beta) \in K$ with $k_u=\l_u'$ for all nodes $u \in V \backslash
   \Set{v}$ and $k_v=0$. For $\l_v'=0$ it holds since \ref{it:NCPx} and \eqref{eqn:NCP} imply
   that $x'_{e}=0$ on all arcs~$e \in \delta^-(v)$. 
   If we define $b(v) \coloneqq \sum_{e \in \delta^+(v)} x_e'- \sum_{e \in \delta^-(v)}x_e'$ for all $v \in V$ the flow $x_e'$
 is a feasible static $b$-transshipment, where $b(v)\geq 0$ for all $v \in V \backslash \Set{t}$. (Note that $s$ has no
   incoming arcs.) Since the graph $G'$ is acyclic and $t$ is the only sink in this $b$-transshipment, we get that
   $\sum_{e \in \delta^+(t)} x_e'= 0$, and therefore the definition of $\Gamma_t$ and \eqref{eqn:VI} imply
   $b(t)\geq -r$. In the following we show that a label of $3M^2$ would induce a flow of $x_e'>r$ on an arc, which is a contradiction.
   Suppose there is a node $w$ with $\l'_w = 3M^2$.
   Since $\l'_s < 3M^2$, there has to be an arc $e=uv$ along an $s$-$w$-path, such that $\l'_u < \l'_v=3M^2$. 
   By \ref{it:NCPx} we can apply \eqref{eqn:NCP} on $\Gamma_e$ to obtain $x_e' \geq \l_v'
   \cdot \capo_e \cdot e^{-\beta_v}
    \geq 3 M^2 \cdot \capo_{\min} \cdot e^{- \log(2M)}
    > M \cdot \capo_{\min} \stackrel{\eqref{eqn:defi_of_M}}{\geq} r$.
   
   Thus, $\l'_v < 3M^2$ for every $v \in V$ and by \eqref{eqn:NCP} it follows flow conservation:
    \begin{equation} \label{eqn:x'_flow_conservation}
   \sum_{e \in \delta^+(v)} x_e'- \sum_{e \in
     \delta^-(v)}x_e' = \begin{cases} r & \text{ if } v = s\\ -r & \text{ if } v = t\\ 0 & \text{ else.} \end{cases}
   \end{equation}
   
   \item Suppose we have $\beta_{v} = \log(2M)$ for some $v \in V$ with $\sum_{e \in \delta^+(v)} x'_e > 0$. For $(x',\l',\gamma)
   \in K$ with $\gamma_u \coloneqq \beta_u$ for all $u \neq v$ and $\gamma_v \coloneqq 0$ we obtain from \eqref{eqn:VI} that
$\l'_{v} \leq \max\limits_{e = v w \in E'} \frac{x'_e}{\b_e^+}$.   
   Let $e_1 =vw$ be an arc that maximizes $\frac{x'_{e}}{\b_{e}^+}$. For $v=s$ we have
     $\l'_s = \frac{1}{e^{-\beta_s}} = 2M \stackrel{\eqref{eqn:defi_of_M}}{>} \frac{r}{\b^+_{\min}} \geq \frac{x'_{e_1}}{\b_{e_1}^+}$, a contradiction.    
   For $v \neq s$ \eqref{eqn:x'_flow_conservation} implies that there is at least one incoming arc $e_2 = uv$ that carries
   $x'_{e_2} \geq \frac{x'_{e_1}}{\abs{\delta^-(v)}} \geq \frac{x'_{e_1}}{\abs{E'}} > 0$ flow. Using \eqref{eqn:NCP} for
   arc $e_2$ yields $\Gamma_{e_2}(x', \l', \beta) = 0$, and therefore we obtain the following contradiction
   \[\l'_v \geq \frac{x'_{e_2}}{\capo_{e_2}\cdot e^{-\beta_v}} 
   \geq \frac{x'_{e_1} \cdot e^{\log(2M)}}{\abs{E'} \cdot\capo_{e_2}} \stackrel{\eqref{eqn:defi_of_M}}{\geq} 
   \frac{x'_{e_1} \cdot 2 \cdot \capo_{\max} \cdot \abs{E'}}{\abs{E'} \cdot \capo_{e_2}  \cdot \b_{\min}^+} 
   > \frac{x'_{e_1}}{\b_{\min}^+} 
   \geq \frac{x'_{e_1}}{\b_{e_1}^+}.\]
   \end{enumerate}
  \end{proof}
   
   \begin{proof}[Proof of \Cref{thm:existence_of_thin_flow}.]
   Let $(x', \tilde \l', \beta)$ be a solution to $\VI(K, \Gamma)$. In
   order to obtain a spillback thin flow we need to make some modifications. Let $V_0 \subseteq V \backslash \set{s}$ be
   the set of nodes with $\sum_{e \in \delta^-(v)} x'_e = \sum_{e \in \delta^+(v)} x'_e  = 0$. We set $\fac_v = 1$ if $v \in V_0$ and $\fac_v = e^{- \beta_v}$ otherwise.   
   Note that we have $\rho_e(\cdot, x'_e, e^{-\beta_v}) = \rho_e(\cdot, x'_e, \fac_v)$ because $\fac_v \neq e^{-\beta_v}$
   implies $x'_e = 0$. Furthermore, let
   \[L \!\coloneqq\! \Set{\!k \!\in\! \R_{\geq 0}^{V} \!|\! k_v \!=\! \tilde \l_v \text{ for } v \in V \backslash V_0 \text{ and } k_v \!\leq\!\! \min_{e = uv \in E'} \rho_e(k_u, x'_e, \fac_v) \text{ for } v \!\in\! V\!}\!.\] 
   Clearly, $\tilde \l' \in L$ since for every $v \in V$ we obtain by
   \eqref{eqn:NCP} applied to $e = uv$ that
   \begin{align*}\tilde \l_v' &\leq \begin{cases} x'_e/(\capo_e \cdot e^{-\beta_v}) & \text{if } e \in E^*\\
   \max \set{\tilde \l'_u, x'_e/(\capo_e \cdot e^{-\beta_v})} & \text{if } e \in E' \backslash E^*
   \end{cases} \\
   &= \rho_e(\tilde \l'_u, x'_e, e^{-\beta_v}) \\
   &= \rho_e(\tilde \l'_u, x'_e, \fac_v).\end{align*}
   So $L$ is non-empty and closed. From the facts that $x'_e$ and $\tilde \l'_s = \frac{1}{e^{-\beta_s}}\leq 2M$ are
   bounded and every node is reachable from $s$ this set is also bounded, i.e., we can define
   $\l' \coloneqq \argmax_{k \in L} \sum_{v \in V} k_v$.
 
   By applying \eqref{eqn:NCP} to the corresponding indices it is easy to check that $(x', \l', \fac)$ indeed satisfies 
   \Cref{eqn:l'_s,eqn:l'_v_min,eqn:l'_v_tight,eqn:l'_v_min_blocked,eqn:l'_v_min_blocked_equal}. The maximality of $\l'$
   in $L$ also guarantees that \eqref{eqn:l'_v_min} is fulfilled for nodes with no in- and outflow.
  \end{proof}

\paragraph{Mixed integer program}
  Spillback thin flows can be computed in practice with a mixed integer program with quadratic constraints.
   In addition to the flow constraints and the conditions \eqref{eqn:l'_s} to
  \eqref{eqn:l'_v_min_blocked_equal} we have to add binary decider variables~$w_e$ for every non-resetting but active arc,
    $y_e$ for every active arc and $z_v$ for every node, where
    \begin{align*}
    w_e=1 \qquad & \Leftrightarrow \qquad  \l_u\geq x'_e/(\capo_e \cdot \fac_v), 
   &&\text{ and thus }\rho_e(x'_e, \l'_u, \fac_v) = \l'_u,\\
    y_e=1 \qquad & \Leftrightarrow \qquad  x_e'=0, 
    &&\text{ and thus \eqref{eqn:l'_v_tight} does not apply,}\\
    z_v=1 \qquad & \Leftrightarrow \qquad  \fac_v=1, 
    &&\text{ and thus \eqref{eqn:l'_v_min_blocked_equal} does not apply.}\qquad \qquad
   \end{align*}
    Since there is no objective function every feasible solution is already a spillback thin flow.

  \paragraph{$\ext$-Extensions.} Let~$\phi \geq 0$ be a fixed point in time. A feasible flow over time with piece-wise
  constant and right-continuous functions~$(f^+, f^-)$ is a \emph{restricted Nash flow over time} on $[0, \phi)$ if it
  is a Nash flow over time for the inflow function~$r_{\phi}(\theta)=r \cdot \mathds{1}_{[0,\phi]}$, where $\mathds{1}$
  is the indicator function. 
  In a Nash flow over time the FIFO principle holds, i.e., no particle entering the network at time~$\theta \geq \phi$
  can influence any particle that has entered the network before time~$\phi$. Thus, all the previous results carry over to
  restricted Nash flows over time. The earliest arrival times~$\l_u(\phi)$ can be determined by taking the left-sided
  limits, which provide us with the current shortest paths network~$G_\phi ' = (V, E'_\phi)$ and the resetting
  arcs~$E^*_{\phi}$. Furthermore, it is possible to determine the spillback arcs~$\bar E_{\phi}$ and the inflow
  bounds~$\b_e^+(\l_u(\phi))$. 
  By \Cref{thm:existence_of_thin_flow} we can obtain a spillback thin flow $(x',\l',\fac)$ on the current shortest paths
  network~$G_\phi '$ with resetting on~$E^*_\phi$ and inflow bounds~$(\b_e^+(\l_u(\phi)))_{e \in E_{\phi}'}$. 
 We set~$x'_e
   \coloneqq 0$ for all~$e \in E \backslash E'_{\phi}$ and extend the
   following functions linearly for some $\ext > 0$:
   \[\l_v(\theta) \coloneqq \l_v(\phi) + (\theta - \phi) \cdot \l'_v \; \text{ and } \;
         x_e(\theta) \coloneqq x_e(\phi) + (\theta - \phi) \cdot x'_e
     \quad \text{ for }\theta \in \:[\phi, \phi + \ext).\]
   Furthermore, the inflow and outflow functions of every arc $e = uv \in E$ are extended by   
    \begin{align*}f^+_e(\theta) &\coloneqq \frac{x'_e}{\l'_u} \quad  \text{for }\theta \in [\l_u(\phi), \l_u(\phi + \ext))\quad\text{ and }\\
    f^-_e(\theta) &\coloneqq \frac{x'_e}{\l'_v} \quad  \text{for }\theta \in [\l_v(\phi), \l_v(\phi + \ext)),\end{align*}
  and the cumulative flow functions~$F_e^+$ and $F_e^-$ are extended accordingly. Note that $\l'_u = 0$ implies that the
  interval~$[\l_u(\phi), \l_u(\phi + \ext))$ is empty, and therefore, $f_e^+$ is not changed in this case. The same is true for $f_e^-$ if $\l_v'=0$. We call the
  family of extended flow functions $(f_e^+, f_e^-)_{e \in E}$ an \emph{$\ext$-extension}.

  \paragraph{Extension step size.} In the following we present some necessary boundaries
  on $\ext$, which we later show to be sufficient for the $\ext$-extension to form a restricted Nash flow over time
  on~$[0, \phi + \ext)$.
%
    Firstly, queues can only deplete until they are empty and, secondly, non active arcs can get active and open alternative routes. Thus, we get the following two conditions on $\ext$ for all $e = uv$:
  \begin{align}
  \l_v(\phi) - \l_u(\phi) + \ext (\l'_v - \l'_u) &\geq \tau_e \;\;\text{ if }e \in E^*_\phi
  \label{eqn:alpha_resetting}\\
  \l_v(\phi) - \l_u(\phi) + \ext (\l'_v - \l'_u) &\leq \tau_e \;\;\text{ if }e \in E \backslash E'_\phi.
  \label{eqn:alpha_others} 
  \end{align}
  In addition, the inflow bounds of the spillback arcs need to be constant within one extension phase, i.e., for all $e = uv \in \bar E_\phi$ we require
  \begin{equation} \label{eqn:alpha_full_arcs}
  \b_e^+(\l_u(\phi) + \theta \cdot \l'_u) = \b_e^+(\l_u(\phi)) \text{ for all } \theta \in [0, \ext).
  \end{equation}
    Finally, the spillback thin flow changes whenever an arc becomes full. Thus, within an extension phase, the total amount of flow on an arc $e = uv \in E'_\phi \backslash \bar E_\phi$ stays strictly under the storage capacity:
  \begin{equation} \label{eqn:alpha_storage}
   d_e(\l_u(\phi + \theta)) < \sigma_e \text{ for } \theta \in [0, \ext).
  \end{equation}
  Note that~$F_e^-$ needs not to be linear on~$[\l_u(\phi), \l_u(\phi + \ext))$.
  We call~$\ext > 0$ \emph{feasible} if it
  satisfies~\Cref{eqn:alpha_resetting,eqn:alpha_others,eqn:alpha_storage,eqn:alpha_full_arcs} and the following lemma shows that such an $\ext$
  always exists.
\begin{lemma}
\label{lem:ext_exists}
For a given restricted Nash flow over time on $[0, \phi)$ there exists a feasible $\ext>0$.
\end{lemma}
\begin{proof}
By \Cref{lem:resetting_implies_active} \ref{it:characterization_of_E'} and \ref{it:characterization_of_E^*} we have
$\l_v(\phi) - \l_u(\phi) > \tau_e$ for $e = uv \in E^*_\phi$ and $\l_v(\phi) - \l_u(\phi) < \tau_e$ for $e = uv \in
E\backslash E'_\phi$. Since $F_e^+(\l_u(\phi)) - F_e^-(\l_u(\phi)) = d_e(\l_u(\phi)) < \sigma_e$ for $e = uv \in
E'_\phi \backslash \bar E_\phi$ we can find an $\ext_1>0$ that satisfies
\Cref{eqn:alpha_resetting,eqn:alpha_others,eqn:alpha_storage}. \Cref{lem:resetting_implies_active}~\ref{it:full_arc_need_time} states that $\l_u(\phi)<\l_v(\phi)$ for full arcs and since $f_e^-$ is piecewise-constant and right-continuous on $[\l_u(\phi), \l_v(\phi))$
so is $\b_e^+$. Hence, there is an $\ext_2 > 0$ satisfying \eqref{eqn:alpha_full_arcs}. Clearly,~$\ext
\coloneqq \min\set{\ext_1, \ext_2}>0$ is feasible.
\end{proof}
   For the maximal feasible $\ext$ we call the interval $[\phi, \phi + \ext)$ \emph{thin flow phase}.
  
  \paragraph{Computing Nash flows over time.} The next theorem shows that it is possible to extend a restricted Nash flow
  over time with spillback step by step using $\ext$-extensions. We cannot hope for a polynomial time algorithm, since
  there are examples with exponential number of thin flow phases, see ~\cite{cominetti2017long}, which means that the output is of exponential size.
  Nevertheless, the constructive nature of the $\ext$-extensions leads to an algorithm which might be
  output-polynomial depending on the computational complexity of a spillback thin flow, which is still an open problem.
  
  \begin{restatable}{theorem}{extensionsarenash}
  \label{thm:extensions_are_Nash}
  Given a restricted Nash flow over time on $[0, \phi)$ and a feasible $\ext>0$, the $\ext$-extension is a restricted
  Nash flow over time on $[0, \phi + \ext)$. Furthermore, the extended $\l$- and $x$-functions are indeed the earliest arrival
  times and the underlying static flows for all $\theta \in [0, \phi + \ext)$.
  \end{restatable}
  To prove this we first show that the $\ext$-extension is a feasible flow over time, where the fair allocation
  condition follows from \eqref{eqn:l'_v_min} and \eqref{eqn:l'_v_tight}, the inflow condition from
  \eqref{eqn:l'_v_min_blocked}, and the no slack condition from \eqref{eqn:l'_v_min_blocked_equal}. Furthermore, the no deadlock condition follows since the total transit time of each cycle is positive. To show that the
  extended $\l$-labels correspond to the earliest arrival times we do a quite technical case distinction. Using this the Nash
  flow condition follows immediately. The formal proof can be found in the appendix.

  \Cref{thm:existence_Nash} finally shows the existence of Nash flows over time in the
  spillback setting.
  \begin{theorem}
  \label{thm:existence_Nash}
  There exists a Nash flow over time with spillback.
  \end{theorem}
  \begin{proof}
  The empty flow over time is a restricted Nash flow over time for the empty set $[0,0)$. For a given restricted Nash
  flow over time $f_i$ on $[0, \phi_i)$ we choose a maximal feasible $\ext_i \in (0, \infty]$, which exists due to \Cref{lem:ext_exists}, and extend $f_i$ with \Cref{thm:extensions_are_Nash} to a restricted Nash flow over time $f_{i+1}$ on $[0, \phi_{i+1})$,
  where~$\phi_{i+1}=\phi_i + \ext_i$. This leads to a strictly increasing sequence~$(\phi_i)_{i \in \N}$. Suppose
  this sequence has a finite limit~$\phi_{\infty} \coloneqq \lim_{i \to \infty} \phi_i<\infty$. In this case we define a
  restricted Nash flow over time~$f^\infty$ for $[0, \phi_\infty)$ by using the point-wise limits of the $x$-  and
  $\l$-functions. 
  Note that the functions remain Lipschitz continuous, and therefore, the process can be continued from this limit point.
  Since this enables us to always extend the Nash flow over time, there cannot be an upper bound on the length of the
  extension interval because the smallest upper bound would correspond to a limit point, which we can extend again.  
 \end{proof}  
Experiments suggest that the number of phases is finite, but we were not able to prove this.

\paragraph{Example.} In \Cref{fig:thin_flow_example} we display the spillback thin flows of the introductory example.
 On the left with $\capo_{e_2}=1$ there are two phases,
and thus, two spillback thin flows. In the second phase $e_3$ becomes active and $e_2$ is resetting. For $\capo_{e_2}=2$  (on the right) there are also two phases. In the second phase $e_2$ becomes full and is therefore
a spillback arc with $\b_{e_2}^+ = 2$. Since $\fac_v = \frac{2}{3} < 1$ arc $e_1$ is throttled to an outflow rate of $2$.
\begin{figure}[ht]
\centering \includegraphics[width=\linewidth]{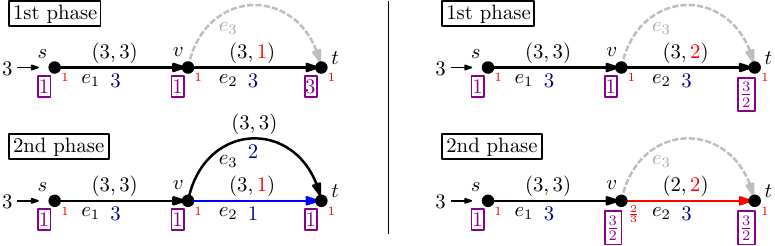} 
\caption{We have $(\b_e^+,\capo_e)$ on top and $x'_e$ on the bottom of each arc, the numbers in the boxes are~$\l'_v$, and the small
numbers are~$\fac_v$. Dashed arcs are non-active and resetting/spillback arcs are blue/red.
} \label{fig:thin_flow_example}
\end{figure}
\section{Kinematic Waves Model} \label{sec:kinematic_wave_model} 
In real traffic situations vehicles cannot immediately enter a fully congested road when someone further down the
street leaves. Instead leaving traffic users create gaps which will be filled after some reaction time by the next
vehicles in line. Hence, it takes time for these gaps to move upstream and only when this free space reaches the entrance of
the road, new cars can enter. This is the key idea of a kinematic wave model and since we consider continuous time and
flow, these gaps will also be represented by a flow over time.
Note hereby, that the speed of the gaps moving upstream is independent of the downstream speed limit (and is in general much lower). 
\Cref{fig:street_model_kinematic_wave} shows how the kinematic waves in a discrete traffic model translate to a
continuous model over time, which we define in the following.
\begin{figure}[ht]
\centering \includegraphics[width=\linewidth]{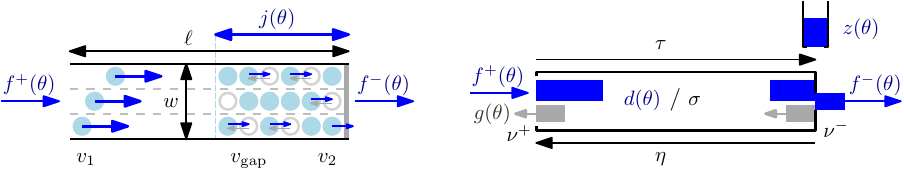} \caption{\emph{Left:} Kinematic wave model of cars using a street
with~$w$~lanes, length~$\l$, speed limit~$v_1$, and exit speed~$v_2$. Gaps within the congestion traverse upstream with a speed
of $v_\text{gap}$. \emph{Right:} The gaps are modeled by a flow over time $g$ traveling upstream with gap transit time
$\eta$. } \label{fig:street_model_kinematic_wave}
\end{figure}

In order to introduce kinematic waves every arc $e=uv$ is equipped with a gap transit time~$\eta_e \geq 0$ and a \emph{gap rate
function}~$g_e: \R \to \R_{\geq 0}$ representing the free spaces between traffic users as a flow
over time traversing upstream from $v$ to $u$. Hereby, $g_e(\theta)$ denotes the gap rate arriving at the tail of the
arc. Note that this gap flow specifies the dynamics of an arc but is no flow through the network, and hence, these
functions do not satisfy flow conservation at the nodes. As gaps are created whenever traffic users leave the link we
set $g_e(\theta) := f_e^-(\theta - \eta_e)$ for $\theta \geq \eta_e$ and $0$ otherwise.
The \emph{total volume of gaps} on the arc is then given by $G_e(\theta) = \int_{\theta}^{\theta+\eta_e} g_e(\xi) \diff
\xi$ .
Since gaps also occupy space on an arc we redefine the \emph{arc load} as $\dwave_e(\theta) \coloneqq F_e^+(\theta) - F_e^-(\theta) + G_e(\theta)$.
Now, we say an arc is full if these redefined arc load reaches the storage capacity. We want that traversing
flow and traversing gaps alone never fill up an arc completely, as this would cause strange pulsing behavior, and furthermore, whenever an arc is full at least some positive gap flow should arrive at the tail. To guarantee this we require the following lower bounds on the storage capacities:
\begin{equation}\label{eqn:storage_condition_wave}
\sigma_e > \capi_e \cdot \tau_e + \max\set{\capi_e, \capo_e} \cdot \eta_e.
\end{equation}
Instead of bounding the inflow rate of a full arc $e = uv$ at time $\theta$ by the outflow at the same time, we
bound it by the gap rate reaching $u$. Hence, we obtain a new definition of the
\emph{inflow bound}:
\[\bwave_e^+(\theta) \coloneqq \begin{cases}
  \min\set{g_e(\theta), \capi_e} & \text{if $e$ is full at time $\theta$,}\\
  \capi_e & \text{else.}
  \end{cases}\]
 As a positive side effect of the kinematic wave model we can relax the no deadlock condition or in most natural
 instances completely remove it: We say a flow over time satisfies the \emph{relaxed no deadlock condition} if at each
 point in time the set of arcs $e$ that are full and have $\eta_e = 0$ is cycle free. Note that, in realistic traffic
 networks $\eta_e$ is positive on every link of positive length (i.e., $\tau_e > 0$), which means that the network does
 not have any directed cycles where all arcs have $\eta_e = 0$, which means that the relaxed no deadlock condition is
 always satisfied.

  Finally we say a flow over time is \emph{feasible} in the kinematic wave model if it satisfies the inflow
  conditions, the fair allocation condition, the no slack condition and the relaxed no deadlock condition,
  where we use $\bwave_e^+$ as inflow bound and we say an arc is full when $\dwave_e(\theta) \geq \sigma_e$.
  
  \begin{remark}
  The kinematic wave model is a generalization of the spillback model, because if we choose all backwards transit times $\eta_e = 0$ then the arc loads, the inflow bounds and even the (relaxed) no deadlock condition coincide in both models.
  \end{remark}

  In order to consider gaps in traffic congestion over multiple arcs we need the following definition.
  A \emph{congestion suffix at time $\theta_1$} is a
  path $(e_1, \dots, e_k)$ such that for all $i \in [k-1]$ we have that $e_i$ is full at time $\theta_i$ with
  $f_{e_i}^+(\theta_i) = \bwave_{e_i}^+(\theta_i)$ and was throttled at time $\theta_{i+1}$, where $\theta_{i+1}
  \coloneqq \theta_i - \eta_{e_i}$. Furthermore, arc~$e_k$ is not full at time $\theta_k$ or was not throttled at time $\theta_k - \eta_{e_k}$, but also has $f_{e_k}^+(\theta_k) = \bwave_{e_k}^+(\theta_k)$. 

\begin{restatable}{lemma}{kinematicwavelemma} \label{lem:kinematic_wave_lemma}
 For a feasible flow over time $f$ in the kinematic wave model we have for all $\theta \in [0, \infty)$:
\begin{enumerate}
\item Storage condition: $\dwave_e(\theta) \leq \sigma_e$. \label{it:storage_condition_wave}
\item If $e$ is full at time $\theta$, we have $z_e(\theta) > 0$. \label{it:arc_saturation_wave}
\item If $e$ is full at time $\theta$, we have $z_e(\theta - \eta_e) > 0$. \label{it:arc_saturation_wave_time_shifted}
\item Every arc that is full at $\theta$ with $f_e^+(\theta) = \bwave_e^+(\theta)$ is part of a (finite) congestion suffix. \label{it:full_arcs_jam_suffix}
\item There is a function $\epsilon\colon [0, \infty) \to (0, 1)$ depending only on the network but not on $f$ such that every arc $e$ with $z_e(\theta)>0$ satisfies $f_e^-(\theta) \geq \epsilon(\theta)$ and $\bwave_e^+(\theta) \geq \epsilon(\theta)$. \label{it:q_well-defined_wave}
\end{enumerate}
\end{restatable}
The proofs of \ref{it:storage_condition_wave} to \ref{it:full_arcs_jam_suffix} follow  straight-forwardly from the
definitions and \ref{it:q_well-defined_wave} is proven similar to \Cref{lem:q_well-defined}. But due to
the relaxed deadlock condition we cannot guarantee an outflow of at least $\epsilon > 0$ with $\epsilon$ depending only
on the network. The formal proofs can be found in the appendix.

\begin{figure}[b]
\centering \includegraphics[width=\linewidth]{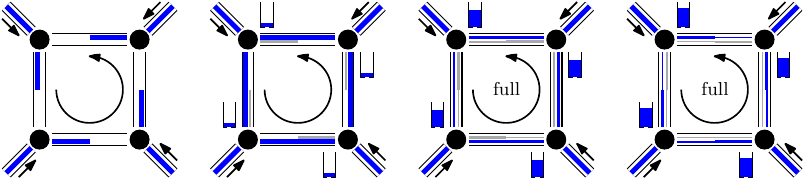} 
\caption{A feasible deadlock at different points in time (from left to right). The gap flow becomes smaller and smaller after the cycle is full, as the outflow rate of an arc in the cycle competes for capacity with the flow rate entering the cycle. For this reason the inflow bound, and hence the inflow rate, also decreases. The flow volume on an arc (without gaps) converges to $\sigma$. The circulating flow converges to $0$ but it remains positive forever. Note that the queues of the non-cycle arcs grow unbounded, which leads to waiting time of $\infty$.
} \label{fig:feasible_deadlock}
\end{figure}

In the kinematic wave model it might be possible that all flow tries to reach a directed cycle and
wants to cycle there forever. In such a \emph{feasible deadlock} (depicted in \Cref{fig:feasible_deadlock}) the cycle is full from some point in time
onward, while the amount of flow in the cycle only converges asymptotically to the total storage capacity but never
reaches it. This means the total amount of gap flow converges to $0$ but always stays strictly positive. Hence,
for all times there is a positive (but decreasing) in- and outflow rate at each arc within the cycle.
Unfortunately, this might cause some
degeneracy since flow waiting to enter the cycle might experience an infinite waiting time.
In the kinematic wave model we set $q_e(\theta) \coloneqq \infty$ if $\Set{ q  \geq 0| \int_{\theta + \tau_e}^{\theta  + \tau_e + q} f_e^-(\xi) \; \diff \xi 
  = z_e(\theta  + \tau_e)} = \emptyset$.
This can happen if $f_e^-(\xi)$ goes quickly to~$0$ for $\xi \to \infty$ such that $\int_{\theta+\tau_e}^{\infty} f_e^-(\xi) \diff \xi < z_e(\theta + \tau_e)$. Consequently, the exit times $T_e$, as well as, the earliest arrival times $\l_e$ can in principle be infinity. Below in \Cref{thm:Nashflows_dont_have_deadlocks} we show that this does not happen for Nash flows over time.

Note that all statements of \Cref{lem:technical_properties} also hold for the kinematic wave model except for Lipschitz property in \ref{it:q_is_lipschitz}. This is mainly used to show that the functions $q_e, \l_v, T_e : [0, \infty) \to [0, \infty]$ are almost everywhere differentiable, which we now show differently in the following lemma.
\begin{lemma} \label{lem:q_diffbar}
For all $e \in E$ the function $q_e$ is differentiable at almost all $\theta$ with $q_e(\theta) < \infty$. The same holds for all $\l_v$ and $T_e$.
\end{lemma}
\begin{proof}
Since by \Cref{lem:technical_properties}~\ref{it:T_monoton} $T_e$ is monotonically increasing on the set $\set{\theta \in [0, \infty) | T_e(\theta) < \infty}$ Lebesgue's theorem for the differentiability of monotone functions states that $T_e$ is almost everywhere differentiable. The same is then true for $q_e(\theta) = T_e(\theta) - \tau_e - \theta$ and $\l_v$ as a minimum of $T_e$ functions.
\end{proof}
Furthermore, \Cref{lem:q'} also holds for almost all $\theta$ with  $q_e(\theta) < \infty$.
To adapt the definition of full arcs from the perspective of some particle $\theta$, we define \[\Ewave_\theta \coloneqq \Set{e = uv \in E |
\dwave_e(\l_u(\theta)) = \sigma_e}.\]
All definitions and statements from \Cref{sec:Nash_flows,sec:constructing_nash_flows}, namely the
definition of a Nash flow over time and of spillback thin flows, as well as \Cref{lem:x_well-defined,lem:ext_exists}, \Cref{thm:existence_of_thin_flow,thm:extensions_are_Nash}, seamlessly translate to the kinematic wave model by
replacing $d_e$ by $\dwave_e$, $\b_e^+$ by $\bwave_e^+$, and $\bar E_\theta$ by $\Ewave_\theta$. There are only a couple of changes, which we discuss in the following.
While all statements of \Cref{lem:resetting_implies_active} transfer to the kinematic wave model, \ref{it:full_subset_active} is proven differently:
\begin{lemma} \label{lem:full_subset_active_wave}
Given a Nash flow over time in the kinematic wave model for all $\theta$ we have $\Ewave_\theta \subseteq E'_\theta$.
\end{lemma}
\begin{proof}
For $e \in \Ewave_\theta$ \Cref{lem:kinematic_wave_lemma}~\ref{it:arc_saturation_wave_time_shifted} states $z_e(\l_u(\theta)-\eta_e) > 0$.
 Therefore, by continuity of $z_e$ and \Cref{lem:kinematic_wave_lemma}~\ref{it:q_well-defined_wave} we have that $f_e^-(\xi) > 0$ for all $\xi \in
 [\l_u(\theta) -\eta_e - \delta, \l_u(\theta) - \eta_e]$ for some small~$\delta > 0$. Considering the amount of flow that has left the arc and whose gap flow also has left the arc we obtain for all $\epsilon > 0$ that
 \begin{align*}F_e^-(\l_u(\theta)) - G_e(\l_u(\theta)) &= \int_0^{\l_u( \theta ) - \eta_e} f^-_e(\xi) \diff \xi \\
 &> \int_0^{\l_u(\theta) - \eta_e - \epsilon} f^-_e(\xi) \diff \xi\\ 
 &= F_e^-(\l_u(\theta) - \epsilon) - G_e(\l_u(\theta) - \epsilon).\end{align*}
 This together with $\dwave_e(\l_u(\theta)) = \sigma_e \geq \dwave_e(\l_u(\theta) - \epsilon)$ yields
 \begin{align*}
 F_e^+(\l_u(\theta)) 
 &= \dwave_e(\l_u(\theta)) + F_e^-(\l_u(\theta)) - G_e(\l_u(\theta))\\ 
 &> \dwave_e(\l_u(\theta)-\epsilon) + F_e^-(\l_u(\theta)-\epsilon)- G_e(\l_u(\theta) - \epsilon)= F_e^+(\l_u(\theta) - \epsilon).
 \end{align*} 
 Hence, \Cref{lem:x_well-defined}~\ref{it:F_increasing_means_active} implies $e \in E'_{\theta}$.
\end{proof}

\Cref{thm:existence_Nash} also holds, because we can again take the point wise limits of the $x$- and $\l$-functions which stay monotone, and therefore, almost everywhere differentiable.
The only thing left to show is, that a Nash flow over time will not produce a deadlock, which would mean infinite waiting times, and therefore, infinite $\l_v$ labels.
\begin{theorem} \label{thm:Nashflows_dont_have_deadlocks}
  Given a Nash flow over time in the kinematic wave model, we have $q_e(\theta), T_e(\theta), \l_v(\theta) < \infty$ for all $\theta \in [0, \infty)$.
\end{theorem}
\begin{proof}
We first show that $\l_v(\theta) < \infty$. Assume for
contradiction there is a minimal $\theta_0$ where some $\l$-label is infinity. Since $\l_s(\theta_0) = \theta_0$ there
has to be some arc $e = uv$ with $\l_u(\theta_0) < \infty = \l_v(\theta_0)$, and thus, $q_e(\l_u(\theta_0)) = \infty$.
This is only possible if $f_e^-(\xi) \to 0$ for $\xi \to \infty$. In other words, for some $\xi_0 > 0$ arc~$e$ would be
throttled for all times $\xi > \xi_0$. Consequently, for each $\xi \geq \xi_0$ there has to be some arc $e' \in \delta_v^+$
that is full due to the no slack condition. 
For some $\epsilon < \min\set{\sigma_{e'} - \capo_{e'} \cdot \eta_{e'} | e' \in E} / r$ we consider the particle $\theta_1
\coloneqq \theta_0 - \epsilon$ for which $\l_t(\theta_1) < \infty$. At time $\l_t(\theta_1)$ all particles $[0,
\theta_1)$ have left the network, and therefore, the amount of flows in the network that is in front of particle $\theta_0$ equals $\epsilon \cdot r$ at this point in time. But this is a contradiction since for every arc~$e' \in \delta_e^+$ that is full at time $\l_t(\theta_1)$ we have
\[\sigma_{e'} = \dwave_{e'}(\l_t(\theta_1)) \leq \epsilon \cdot r + G_{e'}(\l_t(\theta_1)) < \epsilon \cdot r + \capo_e \cdot \eta_e < \sigma_{e'}.\]
Hence, $\l_v(\theta) < \infty$ for all $v \in V$ and all $\theta$. It follows that this also holds for $q_e(\theta)$ and $T_e(\theta)$ for all $e \in E$ and all~$\theta$.
\end{proof}

This theorem ensures that the current shortest path network in every Nash flow over time is acyclic, and thus, \Cref{thm:nash_flow_derivatives_are_thin_flow} translates without any change to the kinematic wave model.

In order to see that this gap flow over time does indeed model the typical kinematic wave phenomena we consider the
following example depicted in \Cref{fig:example_kinematic_wave}. Kinematic waves are best to observe with temporary
bottlenecks such as time-limited lane closures or traffic lights. After the bottleneck at $v_5$ is removed we observe
that the traffic congestion travels upstream and levels out after some time (like surging wave), just as we
would expect it to be in reality. Note that the arcs 
further ahead stay full also at later point in time, even though, the queues decreased and
the congestion moved upstream. But from some point in time onwards the inflow of the preceding arcs are not throttled anymore, and therefore, the flow over time would be exactly the same as if these arcs were not full.

\begin{figure}[ht]
\centering \includegraphics[width=\linewidth]{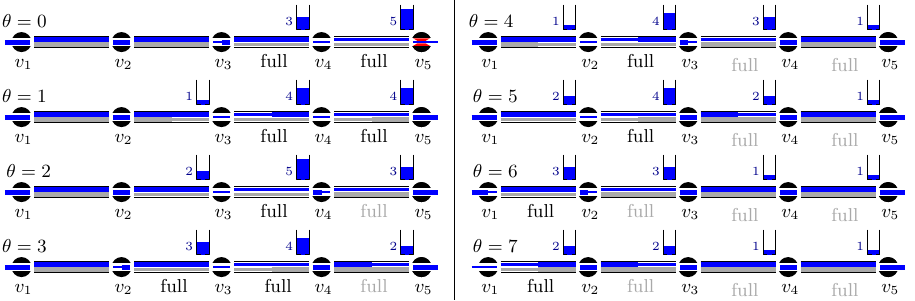} \caption{For every arc $e$ we have $\capi_e = \capo_e = \tau_e =
\eta_e = 2$ and $\sigma_e = 9$. Additionally, we have a constant inflow rate at $v_1$ of $2$ and due to a temporary
bottleneck (red) with capacity $1$ at $v_5$ we start with a congestion between $v_3$ and $v_5$. At $\theta = 0$ the
bottleneck is removed such that the flow leaves $v_5$ with a rate of $2$. Blue lines on top represent the traffic flow
from left to right and gray lines on the bottom are the gaps going from right to left. Thin lines represent a flow with
rate $1$ and fat lines a flow with rate $2$. Arcs that are full, but do not throttle the preceding arcs, are labeled
with a gray ``full''. Within a node we depicted the push through rate.} \label{fig:example_kinematic_wave}
\end{figure}

\begin{remark}
According to \cite{flotterod2016queueing} there is a set of properties that suffices for full consistency with a
kinematic wave model. It is easy to check that our extended model indeed satisfies these conditions when considering
infinitesimal time steps. By taking the limit these conditions describe exactly the derivatives of the arc loads (50.8),
the derivatives of the queues (50.9), the inflow bounds (50.10) and the push rates (50.11). Similar the node properties
are fulfilled as well, since our model implements flow conservation, the FIFO principle, the fair allocation condition
and the no slack condition.
\end{remark}

\section{Relation to the Koch-Skutella-model}
\label{sec:koch_skutella}
\begin{proposition} \label{prop:generalization}
The spillback model, and therefore also the kinematic wave model, is a generalization of the Koch-Skutella-model. If we disable the inflow and the storage capacity, the constructed Nash flow over time in both models coincide.
\end{proposition}

\begin{proof}    
Assume we are given an instance of the Koch-Skutella-model as described in \cite{koch2010nash}, which is a network~$G = (V,E)$ with source~$s$,
sink~$t$ and arcs equipped with transit times~$\tau_e$ and outflow capacities~$\capo_e$. Then, for the spillback model
we keep the network and choose, additionally, storage capacities~$\sigma_e = \infty$ and inflow capacities bigger than the
total outflow capacity of the preceding arcs~$\sum_{e \in \delta^-(u)}\capo_{e}$. This ensures that spillback never
occurs.
Assume for contradiction that there is a node~$v$ and a time~$\theta$ such that the spillback
factor~$\fac_v(\theta)$ is strictly smaller than~$1$. The maximality of $\fac_v(\theta)$ and the fair allocation
condition imply that there is an arc~$e = uv$ with $f_e^-(\theta)<b_e^-(\theta)$, i.e., $e$ is throttled. Due to the no
slack condition there has to be an outgoing arc~$e' = vw$ with $f^+_{e'}(\theta)=b^+_{e'}(\theta)$. This is a
contradiction, because $e'$ can never be full and the inflow capacity is always greater than the inflow rate.

Substituting $\fac_v$ by 1 in the spillback thin flow conditions shows that
\eqref{eqn:l'_v_min_blocked_equal} can be omitted and \eqref{eqn:l'_v_min_blocked} is irrelevant for~$b_e^+$ large
enough. Hence, a spillback thin flow matches a thin flow with resetting as it is stated in the paper 
of~\cite{cominetti2011existence}. Furthermore, we obtain the same bounds on $\ext$ since $\sigma_e = \infty$, and therefore,
the conditions \eqref{eqn:alpha_full_arcs} and \eqref{eqn:alpha_storage} never apply. This shows that a Nash flow over
time with spillback in this network equals a Nash flow over time in the Koch-Skutella-model, and therefore, the spillback
model is indeed a generalization of the Koch-Skutella-model.   
\end{proof}

In the Koch-Skutella-model it was shown that the labels of a thin flow with resetting are unique \cite{cominetti2011existence}, which is not the case with spillback.

\begin{proposition}
A spillback thin flow is not unique, neither is a Nash flow over time with spillback.
\end{proposition}
\begin{proof}
 Consider the example depicted in \Cref{fig:not_uniqueness}. Note that every
convex combination of the displayed spillback thin flows (A and B) is also a valid spillback thin flow. Hence, there is
a continuous amount of Nash flows over time for this instance.
\begin{figure}[t]
\centering \includegraphics[width = \linewidth]{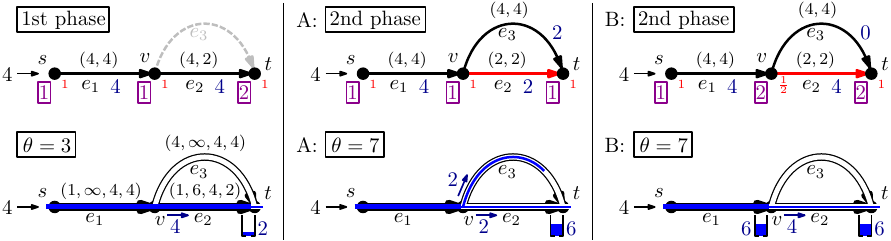} 
\caption{The arc properties (bottom left) are given by $(\tau_e,\sigma_e,
\capi_e,\capo_e)$. For the spillback thin flows (top) we have $(\b_e^+,\capo_e)$ on top and $x'_e$ on bottom right of
each arc, $\l'_v$ in the box under each nodes and the small numbers are~$\fac_v$. Dashed arcs are non-active and
spillback arcs are red. In this instance arc~$e_3$ becomes active exactly at the same time when $e_2$ gets full. This
happens when the particle entering at $\theta = 3$ reaches $v$. There are multiple valid spillback thin flows for the
second phase (bottom, A and B) leading to different Nash flow over time (top, A and B).} \label{fig:not_uniqueness}
\end{figure}
\end{proof}
In the Koch-Skutella-model it is still an open question to find a general bound on the price of anarchy. For some
special cases, some bounds are known; see \cite{bhaskar2015stackelberg,correa2019price}. 
\begin{proposition}
The price of anarchy for Nash flows over time with spillback is unbounded.
\end{proposition}

\begin{proof}
This can be seen by considering a variation of the example in \Cref{fig:example} on the right side. Setting
$\tau_{e_2} = 0$, $\capo_{e_2} =  \sigma_{e_2} = \epsilon$. The Nash flow over time will never use arc $e_3$, and
therefore, the flow arrives with a rate of $\epsilon$. The the social optimum, always
uses $e_3$ in addition, and does not build any queues. Hence, the price on anarchy is larger than $\frac{1}{\epsilon}$.
\end{proof}

\section{Conclusion and Outlook} \label{sec:outlook} 
We extended the Koch-Skutella-model, studied by~\cite{cominetti2011existence,cominetti2015dynamic,koch2010nash}, by spillback and kinematic waves and showed that most of the results achieved for the original model carry over. The existence of equilibria can be shown in a similar way and the structure of a Nash flows over time consists again of a sequence of static flows. 
Furthermore, \cite{cominetti2017long} showed for the Koch-Skutella-model that if the inflow rate~$r$ does not exceed the capacity of a minimal cut, then the lengths
of all queues are bounded. This is not the case in the spillback model, as the example on the right of
\Cref{fig:example} shows. But it is still possible that there exists a phase in the computation of a Nash flow over time
that lasts indefinitely.
It remains open to characterize such long term behavior and to give a bound on the number of
phases. Furthermore, it is a challenging open problem in the original as well as in the spillback model to compute a
thin flow efficiently or to show any hardness results. Since Nash flows over time are intended to describe traffic
situations, a model with multiple origin-destination-pairs would be a huge step. This is difficult since
the earliest arrival times differ for every commodity.

\bibliographystyle{plain}%
\bibliography{literature}

  \newpage
\section{Appendix - Technical Proofs}
\label{sec:appendix}

\technicalproperties*
\begin{proof}\phantom{m}
\begin{enumerate}
\item[\ref{it:fifo}] We have that $q_e(\theta)$ is by definition the first point in time such that
$F_e^-(\theta + \tau_e +q_e(\theta)) -F_e^-(\theta +\tau_e) = z_e(\theta + \tau_e) = F_e^+(\theta) -F_e^-(\theta +\tau_e)$.
Hence, $F_e^-(T_e(\theta))= F_e^+(\theta)$.
\item[\ref{it:q_equiv_z}] This follows directly by the definition of $q_e$.
\item[\ref{it:equal_exit_times}] Intuitively, this holds true since for a particle entering the end of a queue, the entering time does not influence
the time to leave the queue, if no other particle enters the queue in between the two times and if the queue does not
empty out. Formally, this follows by
 \begin{align*}
    q_e(\theta_1) &= \min \Set{ q \geq 0 | \begin{aligned}\int_{\theta_1 + \tau_e}^{\theta_2 + \tau_e} f_e^-(\xi) \diff\xi + &\int_{\theta_2 + \tau_e}^{\theta_1 + \tau_e +q} f_e^-(\xi) \diff\xi \\&= F_e^+(\theta_1) -F_e^-(\theta_1 + \tau_e) \end{aligned}}\\
    &= \min \Set{ p \!=\! q \!-\! \theta_2 \!+\! \theta_1 \!\geq\! 0 | \begin{aligned}&\int_{\theta_2 + \tau_e}^{\theta_2 + \tau_e +p} f_e^-(\xi) \diff\xi 
    \\&= F_e^+(\theta_2) -F_e^-(\theta_2 \!+\! \tau_e) \end{aligned}} \!+\! \theta_2 \!-\! \theta_1\\
    &=q_e(\theta_2) + \theta_2 - \theta_1.
  \end{align*}

Thus, $T_e(\theta_1) = \theta_1 + \tau_e +q_e(\theta_1)=\theta_2 + \tau_e +q_e(\theta_2) = T_e(\theta_2)$.
\item[\ref{it:no_inflow}] To show the contraposition assume $F_e^+(\theta+q_e(\theta))- F_e^+(\theta) > 0$. We have
\[z_e(T_e(\theta))=F_e^+(\theta + q_e(\theta))-F_e^-(T_e(\theta)) \stackrel{\ref{it:fifo}}{=} 
F_e^+(\theta + q_e(\theta)) - F_e^+(\theta) >0.\]
Thus, \Cref{lem:q_well-defined} implies $f_e^-(T_e(\theta)) > \epsilon$.
\item[\ref{it:maxoutflow_depending_on_inflow}] This follows by $z_e(T_e(\theta)) = F_e^+(\theta+q_e(\theta)) - F_e^+(\theta)$ and the definition of $\b_e^-(\theta)$. 
\item[\ref{it:positive_queue_while_emptying}] 
 By definition, $q_e(\theta)$ is the minimal such that $F^-_e(\theta+\tau_e + q_e(\theta))=F_e^+(\theta)$, and therefore $F_e^+(\theta) -
F^-_e(\theta+\tau_e + \xi) > 0$ for $\xi \in [0,q_e(\theta))$. Since $F_e^+$ is monotonically increasing we have for all $\xi
\in [0,q_e(\theta))$ that
\[z_e(\theta + \tau_e + \xi)=F_e^+(\theta+\xi)-F^-_e(\theta+\tau_e + \xi) \geq F_e^+(\theta)-F^-_e(\theta+\tau_e + \xi) > 0.\]

\item[\ref{it:T_monoton}] Consider two points in time $\theta_1 < \theta_2$. Since $F_e^+$ is monotonically increasing, \ref{it:fifo} implies
\begin{equation} \label{eqn:outflow_with_T_monoton}
F_e^-(T_e(\theta_2))=F_e^+(\theta_2)\geq F_e^+(\theta_1) = F_e^-(T_e(\theta_1)). 
\end{equation}
If \eqref{eqn:outflow_with_T_monoton} holds with strict inequality, we obtain by monotonicity of $F_e^-$ that $T_e(\theta_1) < T_e(\theta_2)$. If \eqref{eqn:outflow_with_T_monoton} holds with equality we
have two cases. If $z_e(\theta_2 + \tau_e)>0$, \ref{it:equal_exit_times} states that
$T_e(\theta_1)=T_e(\theta_2)$. If $z_e(\theta_2 + \tau_e)=0$ \ref{it:positive_queue_while_emptying} applied to $\theta_1$ implies $\xi\coloneqq \theta_2-\theta_1 \not\in [0, q_e(\theta_1))$.
Thus, $T_e(\theta_2) \stackrel{\ref{it:q_equiv_z}}{=} \theta_2 + \tau_e \geq \theta_1 + \tau_e + q_e(\theta_1) 
= T_e(\theta_1)$.
That $\l_v(\theta)$ is also monotone follows immediately since it is the minimum of monotone functions.

\item[\ref{it:q_is_lipschitz}] Consider two points in time $\theta_1 < \theta_2$. By \ref{it:T_monoton}, it holds that
$q_e(\theta_2)-  q_e(\theta_1) = T_e(\theta_2) - \theta_2 - T_e(\theta_1) + \theta_1 
\geq \theta_1 - \theta_2$. To obtain an uppper bound we set $\theta_0 \coloneqq \max\set{\theta \in [\theta_1, \theta_2] |
q_e(\theta) = 0}$ if there exists a $\theta \in [\theta_1, \theta_2]$ with
$q_e(\theta) = 0$, and $\theta_0 = \theta_1$ otherwise. With \ref{it:positive_queue_while_emptying}  we get that $z_e(\theta)
> 0$ for all $\theta \in (T_e(\theta_0), T_e(\theta_2))$ and, more importantly,  $q_e(\theta_2) - q_e(\theta_1) \leq
q_e(\theta_2) - q_e(\theta_0)$. 
By \Cref{lem:q_well-defined} we have $F_e^-(T_e(\theta_2)) - F_e^-(T_e(\theta_0)) \geq \epsilon \cdot \left( T_e(\theta_2) - T_e(\theta_0)\right)$, and thus \ref{it:fifo} implies
\begin{align*}
\capi_{e} \left( \theta_2 - \theta_0 \right) &\geq F_e^+(\theta_2) - F_e^+(\theta_0) \\
&= F_e^-(T_e(\theta_2)) - F_e^-(T_e(\theta_0))  \\
&\geq \epsilon \cdot (\theta_2 + q_e(\theta_2)
- \theta_0 - q_e(\theta_0)).
\end{align*}
Finally, we have $q_e(\theta_2) - q_e(\theta_1) \leq q_e(\theta_2) - q_e(\theta_0) \leq \left( \frac{\capi_{e} - \epsilon}{\epsilon} \right) \left( \theta_2 - \theta_0 \right)\leq \left( \frac{\capi_{e} - \epsilon}{\epsilon} \right) \left( \theta_2 - \theta_1 \right)$,
which shows that $q_e$ is
Lipschitz continuous. That $T_e$ and $\l_v$ are Lipschitz continuous follows immediately.
\end{enumerate}
\end{proof}
 
 \xwelldefined*
\begin{proof}
 For the proof of $\ref{it:nash_flow} \Leftrightarrow \ref{it:in_equals_out_at_l}$ see \cite[Theorem 1]{cominetti2015dynamic}.
 
\smallskip $\ref{it:nash_flow} \Rightarrow \ref{it:F_increasing_means_active}$: Suppose $F_e^+(\l_u(\theta) -
 \epsilon) < F_e^+(\l_u(\theta))$ for all $\epsilon > 0$. Since $F_e^+(\l_u(\theta)) > 0$ the Nash flow condition
 implies that $e$ was part of the current shortest paths network at some point in time before $\theta$. Let $\theta'
 \leq \theta$ be the last point in time with $e \in E_{\theta'}'$. Since $e$ was not in the current shortest paths
 network in-between $\theta'$ and $\theta$ there is no inflow during $[\l_u(\theta'), \l_u(\theta)]$, i.e.,
 $F_e^+(\l_u(\theta)) - F_e^+(\l_u(\theta'))=0$. This implies by the assumption that $\l_u(\theta') = \l_u(\theta)$, and
 therefore by \eqref{eqn:bellman} and the monotonicity of $\l_v$ we have
 $\l_v(\theta) \leq T_e(\l_u(\theta))=T_e(\l_u(\theta'))= \l_v(\theta') \leq \l_v(\theta)$.
 Thus, we have equality implying $e \in E'_{\theta}$. 
       
 \smallskip $\ref{it:F_increasing_means_active} \Rightarrow \ref{it:in_equals_out_at_l}$: For $e \in E'_\theta$ we have
 by \Cref{lem:technical_properties}~\ref{it:fifo} that $F_e^+(\l_u(\theta)) = F_e^-(T_e(\l_u(\theta))) =
 F_e^-(\l_v(\theta))$.
 For $e \not \in E'_\theta$, let $\theta_0 \in [0, \theta)$ be minimal with $F_e^+(\l_u(\theta_0)) =
 F_e^+(\l_u(\theta))$, which exists due to the contraposition of \ref{it:F_increasing_means_active}. If $\theta_0 > 0$
 then due to minimality $F_e^+(\l_u(\theta_0) - \epsilon) < F_e^+(\l_u(\theta_0))$ for all $\epsilon > 0$, and therefore
 by \ref{it:F_increasing_means_active} $e$ is active for $\theta_0$. It follows from the observation above, from the
 monotonicity of $F_e^-$ and $\l_v$, as well as, from \Cref{lem:technical_properties}~\ref{it:fifo} that
 \begin{align*}F_e^+(\l_u(\theta)) &= F_e^+(\l_u(\theta_0)) = F_e^-(\l_v(\theta_0)) \\
 &\leq F_e^-(\l_v(\theta)) 
 \leq F_e^-(T_e(\l_u(\theta))) = F_e^+(\l_u(\theta)).\end{align*}
 For $\theta_0 = 0$ we have 
 \[0 \leq F_e^-(\l_v(\theta)) \leq F_e^-(T_e(\l_u(\theta))) = F_e^+(\l_u(\theta)) = F_e^+(\l_u(\theta_0)) = 0.\]
 In both cases we have $F_e^+(\l_u(\theta)) = F_e^-(\l_v(\theta))$.
\end{proof}

\extensionsarenash*
\begin{proof}
    Obviously $f_e^-$ and $f_e^+$ are bounded, piece-wise constant, and right-continuous. All conditions
    are fulfilled on $[0,\phi)$ as well as on $[\phi+\ext,\infty)$ since nothing has changed on this intervals. Note
    that in the first part of the proof we use the linearly extended $\l$-labels and we show only in the end that they
    are indeed the earliest arrival times.
    
    \paragraph{Flow conservation.}
    For $\l'_v > 0$ we obtain for all $v \in V \backslash \set{t}$ and all $\theta \in [\l_v(\phi), \l_v(\phi + \ext))$ that
    \begin{align*}
    \sum_{e \in \delta^+(v)} f_e^+(\theta) - \sum_{e \in \delta^-(v)}f_e^-(\theta) 
    &= \sum_{e \in \delta^+(v)} \frac{x'_e(\theta)}{\l'_v} - \sum_{e \in \delta^-(v)}\frac{x'_e(\theta)}{\l'_v}\\
    &= \begin{cases}
    0 & \text{if } v \in V \backslash \Set{s, t}\\
    r & \text{if } v=s.
    \end{cases}\end{align*} Note that $\l'_s=\frac{1}{\fac_s} = 1$ by the assumptions on arcs in $\delta^+(s)$. For $\l'_v = 0$ we have $[\l_v(\phi), \l_v(\phi + \ext)) = \emptyset$.
  \paragraph{$x$ is well-defined.}
    For all $\xi \in [0, \ext)$ we have
    \begin{equation} \label{eqn:x_well_defined}
    \begin{aligned}
    F_e^+(\l_u(\phi + \xi))&= x_e(\phi) + \int_{\l_u(\phi)}^{\l_u(\phi)+\xi \cdot \l_u'} f_e^+(\theta)\diff\theta
    = x_e(\phi) +\xi \cdot x_e' = x_e(\phi + \xi) \\
    F_e^-(\l_v(\phi + \xi))&= x_e(\phi) + \int_{\l_v(\phi)}^{\l_v(\phi)+\xi \cdot \l_v'} f_e^-(\theta)\diff\theta
        = x_e(\phi) +\xi \cdot x_e' = x_e(\phi + \xi).
    \end{aligned}
    \end{equation}  
  \paragraph{Fair allocation condition.} 
    For every arc $e = uv$ we have to show that \[f_e^-(\l_v(\phi + \xi))=\min \Set{\b_e^-(\l_v(\phi + \xi)), \fac_v \cdot \capo_e}\] for $\xi \in [0, \ext)$.
    This is obvious for $\l'_v = 0$, so we assume $\l'_v > 0$.
  
  \paragraph{Case 1:} $x'_e = 0$.  
    Either $e$ is not active or it is active but $\l'_v > 0$ and \eqref{eqn:l'_v_min} implies that $e$ is not resetting.
    Either way $z_e(\l_u(\phi) + \tau_e) = 0$ and since $f_e^+(\l_u(\phi + \xi)) = 0$ the queue stays empty. We have $f_e^-(\l_v(\phi + \xi)) = \frac{x'_e}{\l'_v} = 0$ and
    $\b_e^-(\l_v(\phi + \xi)) = f_e^+(\l_v(\phi + \xi) - \tau_e) = 0$ for $\xi \in [0, \ext)$ since either $\l_v(\phi + \xi) - \tau_e \geq
    \l_u(\phi)$ (the inflow is part of the current spillback thin flow or even later; in both cases the inflow is
    zero), or $\l_v(\phi + \xi) - \tau_e < \l_u(\phi)$ (the inflow is from earlier than our current spillback thin
    flow). In the later case $e$ is not active for $\zeta$ with $\l_u(\zeta) = \l_v(\phi + \xi) - \tau_e$, since
    \[T_e(\zeta) = \l_u(\zeta) + \tau_e + q_e(\zeta) \geq \l_v(\phi + \xi) > \l_v(\zeta).\] We constructed our flow over
    time on $[0, \phi)$ in such a way that the Nash flow condition is fulfilled for every point in time, and therefore we have
    $f_e^+(\l_v(\phi + \xi) - \tau_e) = f_e^+(\l_u(\zeta))  = 0$ for all $\xi \in [0, \ext)$.
  
  \paragraph{Case 2:} $x'_e >0$ and $e \in E'_\phi \backslash E^*_\phi$ with $\frac{x'_e}{\fac_v \cdot \capo_e} \leq \l'_u$.    
    It follows from (\ref{eqn:l'_v_tight}) that $\l_v' = \l_u'$, and thus $f_e^+(\l_u(\phi + \xi))= \frac{x_e'}{\l_u'} =\frac{x_e'}{\l_v'} =f_e^-(\l_v(\phi + \xi))$ for $\xi \in [0,\ext)$. We obtain
    \begin{align*}
    f_e^+(\l_v(\phi + \xi) - \tau_e) &= f_e^+(\l_v(\phi) - \tau_e + \l'_v \cdot \xi)) \\
    &= f_e^+(\l_u(\phi) + \l'_u \cdot \xi)) \\
    &= f_e^+(\l_u(\phi + \xi))  \\
    &= f_e^-(\l_v(\phi + \xi)).\end{align*}
   This equality yields $z_e(\l_v(\phi + \xi))=z_e(\l_v(\phi)) + \int_{\l_v(\phi)}^{\l_v(\phi+\xi)}
   f_e^+(\zeta - \tau_e) - f_e^-(\zeta) \diff \zeta= 0$. By the case distinction we have \[\b_e^-(\l_v(\phi + \xi)) =
   f_e^+(\l_v(\phi + \xi) - \tau_e) = \frac{x_e'}{\l_u'} \leq \fac_v \cdot \capo_e.\] In conclusion we have
    \[\min \!\Set{\!\b_e^-(\l_v(\phi + \xi)), \fac_v \!\cdot\! \capo_e\!} \!=\! \b_e^-(\l_v(\phi + \xi)) 
    \!=\! f_e^+(\l_v(\phi + \xi) - \tau_e) = f_e^-(\l_v(\phi + \xi)).\]

  \paragraph{Case 3:} $x'_e >0$ and ($e \in E^{*}_{\phi}$ or $e \in E'_\phi \backslash E^*_\phi$ with $\frac{x'_e}{\fac_v \cdot \capo_e} > \l'_u$).  
   It follows from (\ref{eqn:l'_v_tight}) that $\l'_v= \frac{x'_e}{\fac_v \cdot \capo_e}$, and thus $f_e^-(\l_v(\phi+\xi))=
   \frac{x'_e}{\l'_v}= \fac_v \cdot \capo_e$ for $\xi \in [0, \ext)$. It remains to show that $\b_e^-(\l_v(\phi + \xi))\geq
   \fac_v \cdot \capo_e$. For $e \in E^*_{\phi}$ we get from \eqref{eqn:alpha_resetting} that $\l_v(\phi)-\l_u(\phi) +
   \xi \cdot (\l'_v - \l'_u) > \tau_e$ for $\xi \in [0,\ext)$. For $e \in E'_\phi \backslash E^*_\phi$ and $\l'_v =
   \frac{x'_e}{\fac_v \cdot \capo_e} > \l'_u$ it follows that $\l_v(\phi)-\l_u(\phi) = \tau_e$ and $\xi \cdot (\l'_v - \l'_u)
   > 0$ for $\xi \in (0,\ext)$. In both cases we get that $\l_v(\phi + \xi) - \tau_e > \l_u(\phi + \xi)$ for $\xi \in
   (0,\ext)$. It follows with the monotonicity of $F_e^+$ that
   \begin{align*}
   z_e(\l_v(\phi + \xi))&\stackrel{\eqref{eqn:x_well_defined}}{=}F_e^+(\l_v(\phi + \xi)- \tau_e)- F_e^+(\l_u(\phi + \xi))\\
   &\;\geq F_e^+(\l_u(\phi + \xi)+ \varepsilon)- F_e^+(\l_u(\phi + \xi)) \\
   &\;= \varepsilon \cdot \frac{x'_e}{\l'_u} > 0,
   \end{align*}
   where we choose $\varepsilon>0$, such that $\l_u(\phi \!+\! \xi)+ \varepsilon < \min \set{\l_u(\phi \!+\!\ext), \l_v(\phi \!+\!
   \xi) \!-\! \tau_e}$.
   Note that since a flow of $x'_e$ leaves node $u$ there either has to be some inflow of $x'$ into $u$ or $u = s$. In
   both cases we have $\l'_u > 0$, and thus $\l_u(\phi + \xi) < \l_u(\phi +\ext)$ and $\frac{x'_e}{\l'_u}$ is well-defined. Finally, $\b_e^-(\l_v(\phi + \xi)) = \capo_e \geq \fac_v \cdot \capo_e$.

  \paragraph{Inflow condition and no slack condition.} For all $\xi \in [0, \ext)$ we show that $f_e^+(\l_u(\phi + \xi))
  \leq \b_e^+(\l_u(\phi + \xi))$ and that it holds with equality for at least one arc $e \in \delta_v^+$, whenever there is an incoming throttled arc.
  Equation~\eqref{eqn:alpha_storage} ensures that arcs $e \not \in \bar E_\phi$ stay non-full during $[\l_u(\phi),
  \l_u(\phi + \ext))$. Together with \eqref{eqn:alpha_full_arcs} we get that $\b_e^+(\l_u(\phi + \xi)) = \b_e^+$ for all
  $\xi \in [0, \ext)$, and hence \eqref{eqn:l'_v_min_blocked} yields \[f_e^+(\l_u(\phi + \xi)) = \frac{x'_e}{\l'_u} \leq \b_e^+ = \b_e^+(\l_u(\phi + \xi)).\] 
  An incoming throttled arc implies $\fac_u<1$, and thus the inequality
  holds due to \eqref{eqn:l'_v_min_blocked_equal} with equality.

    \paragraph{No deadlock condition.} 
    Suppose there is a point in time $\xi$ when the set of full arcs contain a cycle
    $v_1, \dots, v_k = v_0$. For every $i = 0,1, \dots, k$ we consider the minimal value $\theta_i$ such that $\l_{v_i}(\theta_i) =
    \xi$. By  \Cref{lem:resetting_implies_active}~\ref{it:full_subset_active}
    we have 
    \[\l_{v_i}(\theta_{i-1}) -  \tau_{v_{i-1}v_i} \geq \l_{v_{i-1}}(\theta_{i-1}) = \xi = \l_{v_i}(\theta_i)\] 
    for every $i = 1, \dots, k$, 
    which implies $\theta_{i-1} \geq \theta_i$. Since the sum of transit times in each cycle is strictly
    positive there has to be an~$i$ with $\tau_{v_{i-1}v_i} > 0$, and therefore $\theta_{i-1} > \theta_i$, which leads to a contradiction.
    
   \paragraph{Earliest arrival times.} We show that the extended $\l$-labels fulfill
   \Cref{eqn:bellman}, and therefore describe the earliest arrival times. As shown before we have $\l'_s = 1$
   implying $\l_s(\theta) = \theta$ for all $\theta \in [0, \phi + \ext)$. Considering $v \neq s$, $e = uv \in E$, and
   $\xi \in [0, \ext)$, we distinguish two cases
   and show $\l_v(\phi + \xi) \leq T_e(\l_u(\phi + \xi))$ in the first
   case and $\l_v(\phi + \xi) = T_e(\l_u(\phi + \xi))$ in the second case.
   
   \paragraph{Case 1:} $e \in E\backslash E'_\phi$ or $e \in E'_\phi \backslash E^*_\phi$ with $\l'_v < \l'_u$. \\
     We have for all $\xi \in [0, \ext)$ that
     \[\l_v(\phi + \xi) \leq \l_u(\phi) + \tau_e + 
     \xi \cdot \l'_u \leq T_e(\l_u(\phi)  + \xi \cdot \l'_u) = T_e(\l_u(\phi + \xi)),\]
     where the first inequality follows by \eqref{eqn:alpha_others} for $e \in E\backslash E'_\phi$ or by 
     $\l_v(\phi) = \l_u(\phi) + \tau_e$ and $\l'_v < \l'_u$ otherwise.
       
    \paragraph{Case 2:} $e \in E^*_\phi$ or $e \in E'_\phi \backslash E^*_\phi$ with $\l'_v \geq \l'_u$.\\  
    If $x'_e = 0$ and $e \in E^*_\phi$ we get from \eqref{eqn:l'_v_min} that $\l'_v = \rho_e(\l'_u, x'_e, \fac_v) = 0$.
    Since $e$ is active for $\phi$ it follows that
    \[\l_v(\phi + \xi) = \l_v(\phi) = T_e(\l_u(\phi)) \leq T_e(\l_u(\phi + \xi)).\]
    To show equality note that with \eqref{eqn:alpha_resetting} we have \[q_e(\l_u(\phi + \xi))=T_e(\l_u(\phi +\xi)) - \l_u(\phi + \xi) - \tau_e >
     T_e(\l_u(\phi +\xi)) - \l_v(\phi +\xi) \geq 0.\]
    Thus, \Cref{lem:technical_properties}~\ref{it:q_equiv_z} and \ref{it:no_inflow} together with \[F_e^+(\l_u(\phi+\xi))-F_e^+(\l_u(\phi))= \xi \cdot
    x'_e = 0\]
    imply $T_e(\l_u(\phi)) = T_e(\l_u(\phi + \xi))$.
    
    If we have $x'_e = 0$ and $e \in E'_\phi \backslash E^*_\phi$ with $\l'_v \geq \l'_u$ we obtain \[\l'_v \leq
    \rho_e(\l'_u, x'_e, \fac_v) = \l'_u \leq \l'_v,\] and hence $\l'_v = \l'_u$. This yields
    \[\l_v(\phi + \xi) = \l_u(\phi) + \tau_e + 
    \xi \cdot \l'_u = \l_u(\phi + \xi) + \tau_e = T_e(\l_u(\phi + \xi)),\]
    where the last equality holds since there is no inflow within $(\l_u(\phi), \l_u(\phi + \xi))$, and therefore no
    queue.
    
   Now, suppose that $x'_e > 0$, which implies $\l'_v = \rho_e(\l'_u, x'_e, \fac_v) > 0$.
    For all $\xi \in [0, \ext)$ we have
    \begin{equation}\label{eqn:l_u_plus_tau}\l_u(\phi + \xi) + \tau_e = \l_u(\phi) + \tau_e + \xi \cdot \l'_u 
    \leq \l_v(\phi) + \xi \cdot \l'_v = \l_v(\phi + \xi),\end{equation}
    where the inequality follows either from $\eqref{eqn:alpha_resetting}$ in the case of $e \in E^*_\phi$ or, in the
    other case, by $\l_v(\phi) = \l_u(\phi) + \tau_e$ and $\l'_u \leq \l'_v$. By definition of $q_e$ and $z_e$ we obtain
    that $q_e(\l_u(\phi + \xi))$ is the minimal non-negative value with
    \[F_e^-(\l_u(\phi + \xi) + \tau_e + q_e(\l_u(\phi + \xi))) = F_e^+(\l_u(\phi + \xi)) \stackrel{\eqref{eqn:x_well_defined}}{=} F_e^-(\l_v(\phi + \xi)).\]
    Note that $F_e^-$ is monotone and strictly increasing at $\l_v(\phi + \xi)$ with slope
    $f_e^-(\l_v(\phi + \xi)) = \frac{x'_e}{\l'_v} > 0$. This and \eqref{eqn:l_u_plus_tau} imply that $q_e(\l_u(\phi +
    \xi))$ satisfies
    \[T_e(\l_u(\phi + \xi)) = \l_u(\phi + \xi) + \tau_e + q_e(\l_u(\phi + \xi)) = \l_v(\phi + \xi).\]   
    
    In conclusion, both cases together show that for all $v \in V \backslash \Set{s}$ and all $\xi \in [0, \ext)$ we have
    \[\l_v(\phi + \xi) \leq \min_{e = uv \in E} T_e(\l_u(\phi + \xi)).\]
    In order to show equality recall that \eqref{eqn:l'_v_min} yields an arc $e = uv \in E'$ with $\l'_v =
    \rho_e(\l'_u, x'_e, \fac_v)$. Hence, either $e \in E^*$ or $\l'_v \geq \l'_u$, i.e., $e$ belongs to the second case and there we showed equality.

    \paragraph{Nash flow condition.} Since all conditions are fulfilled, we have a feasible flow over time. The Nash
        flow condition follows immediatly by
    \Cref{lem:x_well-defined} \ref{it:in_equals_out_at_l} and \eqref{eqn:x_well_defined}. Note that by construction, the condition holds for every point in time and not only for
    almost every point in time.
   \end{proof}

\kinematicwavelemma*
\begin{proof}
\begin{enumerate}
\item[\ref{it:storage_condition_wave}] Assume for contradiction that $\dwave_e(\theta)>\sigma_e$ at some point. Since
$\dwave_e$ is continuous and $\dwave_e(0) = 0$ we find an interval $(\theta_0, \theta]$ with $\dwave_e(\theta_0) =
\sigma_e$ and $\dwave_e(\theta) > \sigma_e$ for all $\theta \in (\theta_0, \theta]$. From the inflow condition it
follows that $f_e^+(\theta)\leq g_e(\theta)$ for all $\theta \in [\theta_0, \theta]$. This leads to a
contradiction, since
\begin{align*}0 < \dwave_e(\theta) - \dwave_e(\theta_0) &= \int_{\theta_0}^{\theta} f_e^+(\xi) - f_e^-(\xi) + g_e(\xi + \eta_e) - g_e(\xi) \diff \xi \\
&= \int_{\theta_0}^{\theta} f_e^+(\xi)- g_e(\xi) \diff \xi \leq 0.\end{align*}  

\item[\ref{it:arc_saturation_wave}]
 The inflow condition yields $f_e^+(\theta) \leq \capi_e$ and by definition we have $g_e(\theta) \leq \capo_e$. Hence, \begin{align*}z_e(\theta) &= F_e^+(\theta - \tau_e) - F_e^-(\theta) \\
 &\geq F_e^+(\theta) - \capi_e \cdot \tau_e - F_e^-(\theta) + G_e(\theta) - \capo_e \cdot \eta_e\\ &\stackrel{\eqref{eqn:storage_condition_wave}}{>} \dwave_e (\theta) - \sigma_e \geq 0.\end{align*}
 
\item[\ref{it:arc_saturation_wave_time_shifted}] Again with $f_e^+(\theta) \leq \capi_e$ and the definitions of $z_e$,
$G_e$ and $\dwave_e$ we obtain:
\begin{align*}z_e(\theta - \eta_e) &= F_e^+(\theta - \eta_e - \tau_e) - F_e^-(\theta - \eta_e)\\
&\geq F_e^+(\theta) - \capi_e \cdot (\tau_e + \eta_e) - F_e^-(\theta)+ G_e(\theta)\\ &\stackrel{\eqref{eqn:storage_condition_wave}}{>} \dwave_e (\theta) - \sigma_e \geq 0.\end{align*}

\item[\ref{it:full_arcs_jam_suffix}]
 If $e$ is not throttled at time $\theta - \eta_e$ we are done, since $(e)$ is a finite congestion suffix.  So suppose $e$ is throttled at time $\theta_2 = \theta - \eta_e$ then, by the no
slack condition, there has to be a consecutive arc $e_2$ with $f_{e_2}^+(\theta_2) = \bwave_{e_2}^+(\theta_2)$. If $e_2$ is full at time $\theta_2$ and was throttled at time $\theta_3 \coloneqq \theta_2 - \eta_{e_2}$ we find an arc $e_3$ with $f_{e_3}^+(\theta_3) = \bwave_{e_3}^+(\theta_3)$. We repeat this until we obtain an arc $e_k$ that has $f_{e_k}^+(\theta_k) = \bwave_{e_k}^+(\theta_k)$ but is not full at time $\theta_k$ or not throttled at time $\theta_k - \eta_{e_k}$. To show that such an arc $e_k$ exists, assume for contradiction that there is an infinite sequence $(e_1,
e_2, e_3, \dots)$. Since $E$ is finite, there is some node $v$ which is visited infinitely often. But due to the relaxed no deadlock condition each time we visit $v$ we
consider a point in time which is strictly earlier than the last visit of $v$. Since no arc is full at time $0$ this is a contradiction.

\item[\ref{it:q_well-defined_wave}] The proof is similar to the proof of \Cref{lem:q_well-defined}. We set
\begin{align*}
\capmin&\coloneqq \min \left(\set{ \capi_e, \capo_e | e \in E } \cup \set{1}\right),&& \Sigma \coloneqq \max\set{\sum_{e
\in E} \capo_e, 1}, \\
\etamin &\coloneqq \min \left(\set{\eta_e > 0| e \in E} \cup \set{1}\right) \; \text{ and}\!\!  &&\epsilon(\theta) \coloneqq \left(\frac{\capmin}{\Sigma}\right)^{\abs{E} \cdot \frac{\theta}{\etamin}} \cdot \capmin.
\end{align*}
 If $e$ is not
throttled at time $\theta$, we have $f_e^-(\theta) = \capo_e$, so suppose $e$ is throttled. By the no slack condition there has to be
a consecutive arc $e_1$ with $f_{e_1}^+(\theta) = \bwave_{e_1}^+(\theta)$. Due to \ref{it:full_arcs_jam_suffix} there has to be a congestion suffix $(e_1, e_2, \dots, e_k)$ at
time $\theta_1 = \theta$ where $k \leq \frac{\theta}{\etamin} \cdot \abs{E}$. If
$e_k$ is not full we have $f_{e_k}^+(\theta_k) = \bwave_{e_k}^+(\theta_k) = \capi_{e_k}$. If $e_k$ is full but not
throttled at time $\theta_{k+1} \coloneqq \theta_k - \eta_{e_k}$ we have by \ref{it:arc_saturation_wave_time_shifted}
that $e_k$ had a queue at time $\theta_{k+1}$, and therefore $g_{e_k}(\theta_{k}) = \capo_{e_k}$, which leads to
$f_{e_k}^+(\theta_k) = \bwave_{e_k}^+(\theta_k) = \min\set{\capi_{e_k}, \capo_{e_k}}$. Furthermore, for two consecutive
arcs $e_{i} = uv$ and $e_{i+1} = vw$ we have:
  \begin{equation} \label{eq:lowerbound_on_outflow_wave} \begin{aligned} f_{e_i}^-(\theta_{i+1}) &= \fac_v(\theta_{i+1}) \cdot \capo_{e_i} \\
  &\geq \frac{\sum_{e' \in \delta^+(v)}f_{e'}^+(\theta_{i+1})}{\sum_{e' \in \delta^-(v)}\capo_{e'}} \cdot \capmin \\
  &\geq \frac{f_{e_{i+1}}^+(\theta_{i+1})}{\Sigma}\cdot \capmin. \end{aligned}\end{equation}
  Since the arc $e_i$ is full at time $\theta_i$ with exhausted inflow capacity it holds that
  $f_{e_i}^+(\theta_i)=\bwave^+_{e_i}(\theta_i) = \min\set{f_{e_i}^-(\theta_{i+1}), \capo_{e_i}}$. Recursive application
  of \eqref{eq:lowerbound_on_outflow_wave} along the sequence gives $f_e^-(\theta)\geq
  \left(\frac{\capmin}{\Sigma}\right)^k \cdot \capmin \geq \epsilon(\theta)$.   
With \ref{it:arc_saturation_wave_time_shifted} we have $\bwave_e^+(\theta)
  = \min\set{f_e^-(\theta - \eta_e), \capi_e} \geq \epsilon(\theta - \eta_e) \geq \epsilon(\theta)$ for full $e$.
\end{enumerate}
\end{proof}


%
%



\end{document}